\documentclass[11pt]{article}
\usepackage{graphicx} 

\usepackage{amsthm}
\usepackage[T1]{fontenc}
\usepackage{pgfplots}
\usepackage{pgfplotstable}
\pgfplotsset{compat=1.8}
\usepackage[aboveskip=8pt,belowskip=-4pt]{caption}
\usepackage{subcaption}
\usepackage{xspace} 
\usepackage[margin=1in]{geometry}
\usepackage{textcomp}
\usepackage{balance}  
\usepackage{xcolor}
\usepackage{amsmath}
\usepackage{bbm}
\DeclareMathOperator*{\argmax}{argmax}
\usepackage{algorithm}
\usepackage{algorithmicx}
\usepackage{cite}
\usepackage{float}
\usepackage{siunitx}
\usepackage{graphicx}
\usepackage{xspace}
\usepackage{thm-restate}
\usepackage[disable]{todonotes}
\usepackage{balance}  
\usepackage{booktabs} 

\usepackage{marginnote} 
\usepackage{url}
\usepackage{amssymb}
\usepackage{enumitem}
\graphicspath{{./fig/}}
\usepackage{mathtools}
\usepackage[noend]{algpseudocode}
\newtheorem{theorem}{Theorem}
\newtheorem{lemma}{Lemma}

\newtheorem{definition}{Definition}

\newtheorem{proposition}{Proposition}
\usepackage{comment}
\usepackage{color}
\usepackage{cleveref}
\usepackage{tikz}
\usetikzlibrary{fadings}
\usepackage{MnSymbol,wasysym} 

\usepackage{titling}
\thanksmarkseries{arabic}

\newcommand{\defn}[1]       {{\textit{\textbf{\boldmath #1}}}}
 
\newcommand{\poly}{\operatorname*{poly}}
\newcommand{\polylog}{\operatorname*{polylog}}
\newcommand{\E}{\mathbb{E}}

\newenvironment{restatedtheorem}[1]
  {\par\medskip\noindent\textbf{Theorem~#1.}\itshape\hspace{0.5em}\ignorespaces}
  {\par\medskip\noindent\normalfont\ignorespacesafterend}

\newif\ifcomments
\commentsfalse

\newcommand{\thmmainbody}{%
    Let $f: [N] \to [N]$ be a function which we have $O(1)$-time oracle access to, and that is updated over time by $\mathsf{Update}(x, y)$ operations. For $T = O(\log N / \log \log N)$, we can construct in linear time a data structure that supports $O(T)$-time iterators, that uses space $O(N \log N / T)$ bits, and that supports $O(T)$-time $\mathsf{Update}(x, y)$ operations. Both the initial construction and each $\mathsf{Update}$ operation succeed with high probability in $N$.%
}

\newcommand{\thmunorderedgraphsbody}{%
    Let $\epsilon \in (0, \frac{1}{2})$ satisfy $\epsilon^{-1}=O(\log n/\log\log n)$. Let $G=([n],E)$ be a graph that is updated over time by edge insertions and
    deletions. There is a $(1+\epsilon)$-compact dynamic graph data structure that, given the initial edge set $E$, can be constructed in $O(n+m)$ time; that supports adjacency queries in $O(1)$ time, insertions and deletions in $O(\epsilon^{-1})$ amortized time, and iterator queries in $O(\epsilon^{-1})$ time per neighbor, where the processing/insert/delete times are all with high probability in $n$.%
}

\newcommand{\thmfiatnaorbody}{%
    Let $f: [N] \to [N]$ be a function which we have $O(1)$-time oracle access to, and that is updated over time by $\mathsf{Update}(x, y)$ operations. For $t \in O(N^{1/3})$, we can construct in $\tilde{O}(N)$ time a data structure that supports $\tilde{O}(t^3)$-time iterators, that uses space $O(N \log N / t)$ bits, and that supports $O(t^3)$-time $\mathsf{Update}(x, y)$ operations. Both the initial construction and each $\mathsf{Update}$ operation succeed with high probability in $N$.%
}

\title{Dynamic and Optimal Function Inversion in the Small-Time Regime}

\author{John Kuszmaul\thanks{MIT CSAIL. Email: john.kuszmaul@gmail.com. Research supported by NSF Grant CCF-2330048, BSF Grant 2024233, and a Simons Investigator Award.} \and William Kuszmaul\thanks{CMU. Email: kuszmaul@cmu.edu. Research supported in part by NSF grant CCF-2504471 and by a Jane Street grant.}}
\date{}

\begin{document}

\maketitle

\begin{abstract}
The classic function-inversion problem considers the task of constructing a data structure which, given access to a constant-time oracle for a function $f : [N] \rightarrow [N]$, supports efficient inverse-queries on $f$. This problem has been studied extensively in the small-space/large-time regime, where one wishes to use space $S$, say, $N^{1 - \Omega(1)}$ bits, and where the query time is intended to be a small polynomial of $N$. Much less attention has been given to the \emph{small-time/large-space} regime, where $S = (N \log N) / t$ for some relatively small $t$, and where the goal is to achieve a good space bound as a function of $t$. 

In this paper, we give an optimal solution in the small-time regime, achieving space $S = O(N \log N / t)$ and time $O(t)$ for any $t \le O(\log N / \log \log N)$. This matches a lower bound by Yao (and is the first parameter regime where the lower bound has been matched for general functions). Additionally, we extend our solution to support point-updates to $f$, also in $O(t)$ time. Our techniques for supporting point updates also extend to the classic function-inversion solution of Fiat and Naor. 

All of our results are motivated by the data-structural perspective on function inversion, in which the goal is to supplement an already-existing data structure $\mathcal{D}_1$ (which, as part of its functionality, encodes some function $f$) with a small secondary data structure $\mathcal{D}_2$ that supports inverse queries. Our results allow $\mathcal{D}_2$ to be implemented in $(N \log N)/t$ bits with $O(t)$ query (and update) times -- if $\mathcal{D}_1$ is itself $\Theta(N \log N)$ bits, this results in the overall space usage increasing by only a $(1 + O(1/t))$ factor. 

As a sample application of our results, we show how to construct dynamic unordered graphs that use space $(1 + \epsilon)$-close to information-theoretically optimal while offering adjacency queries, neighborhood queries, and edge insertions/deletions in amortized time $O(\epsilon^{-1})$.
\end{abstract}

\section{Introduction}

The function-inversion problem considers the following basic question. Suppose we are given constant-time black-box access to some function $f: [N] \rightarrow [N]$, and that we wish to construct a data structure $\mathcal{D}$ that answers queries of the form 
$$
\mathsf{Inverse}(y) = \begin{cases}
x & \text{ for some } x \in f^{-1}(y) \text{ if } |f^{-1}(y)| > 0, \\
\bot & \text{ otherwise.}
\end{cases}
$$
Given a space budget of $S = (N \log N) / t$ bits, for some $S$, how efficiently can we hope to answer such queries?

On the upper-bound side, a classic result of Fiat and Naor \cite{FiatNa91, FiatNa00} achieves $\tilde{O}(t^3)$-time inverse queries for any function $f$, extending an earlier result of Hellman \cite{Hellman80} who achieved time $\tilde{O}(t^2)$ for a random function.

These early constructions \cite{Hellman80, FiatNa91, FiatNa00} were motivated by applications to cryptography, in which one wishes to build a data structure $\mathcal{D}$ that can speed up the task of inverting some cryptographic hash function $f$. Indeed, a variation of Hellman's original construction, known as a \emph{rainbow table} \cite{oechslin2003making}, has been widely used as a practical technique to crack passwords \cite{thing2009novel, meyer2011breaking,Losby2026NetNTLMv1RainbowTables}, and is one of the reasons that password salting has become a widely-accepted security practice \cite{CrackStation2021SaltedPasswordHashing,Arias2025AddingSaltToHashing}.\footnote{It turns out that the bottleneck in applying Hellman's construction (or rainbow tables) is the \emph{domain size} of $f$ (rather than the co-domain size). Without salting, it is possible to restrict the domain of $f$ to a smaller size (say, short strings only, or some set of highly structured inputs) in order to make the construction practical even against high-quality cryptographic hash functions.} 

On the lower-bound side, Yao proved in 1990 \cite{Yao90} that any solution must incur $\Omega(t)$-time queries. Despite a great deal of interest, it is not known whether Yao's bound is achievable (or how to get any non-trivial results for $S \ll N^{1/2}$ \cite{golovnevGuPe23}). (It is, however, known that any $t^{2 - \Omega(1)}$-time solution would require a different class of data structure than the one considered by Fiat and Naor \cite{barkan2006rigorous}.)

Cryptographic applications of function inversion naturally tend to focus on what we will call the \emph{small-space/large-time regime}, where the goal is to use space $N^{1 - \Omega(1)}$ bits, while allowing for a query time of the form $N^{\Theta(1)}$. In this paper, we consider instead the \emph{small-time/large-space regime}, where the goal is to support a query time as small as possible, while also ensuring that the space usage is $o(N \log N)$ bits. 

This is motivated in part by the following basic data-structural setting. Suppose we are maintaining a $\Theta(N \log N)$-bit data structure $\mathcal{D}_1$ that either implicitly or explicitly encodes a function $f: [N] \rightarrow [N]$, and we want to \emph{add} a secondary data structure $\mathcal{D}_2$ to support inverse queries on $f$. A natural goal here is to design $\mathcal{D}_2$ both to be very time efficient (hence the small-time regime) while also using space $o(N \log N)$ bits (so that \emph{overall} space efficiency is preserved, up to $1 + o(1)$ factors). This raises a natural question: can we, in this low-time regime, hope to do better than the classic Fiat and Naor construction?

Of course, in data structural settings, the function $f$ being considered is often itself a \emph{dynamically changing object}. This raises a second question, which is interesting in all parameter regimes: Can we build function-inversion data structures that support efficient point-updates to $f$, while still supporting efficient function inversion, \emph{even as $f$ changes over time?}

\paragraph{This paper. }In this paper, we construct an optimal data structure for the small-time/large-space regime. For any $t \in O(\log N / \log \log N)$, we are able to construct a solution with space $O(N \log N / t)$ bits, while supporting $O(t)$-time inversion queries (thereby matching Yao's lower bound). Moreover, our solution is dynamic, allowing point updates to $f$ to \emph{also} be implemented in $O(t)$ time. These guarantees are captured in the following theorem, which is proven in Section \ref{sec:main}:

\begin{restatedtheorem}{\ref{thm:main}}
    \thmmainbody
\end{restatedtheorem}

Our solution can be viewed as an extension of the basic chaining idea introduced by Hellman \cite{Hellman80} and Fiat and Naor \cite{FiatNa91,FiatNa00}, but with significant data-structural machinery to speed up both the query time, the construction time, and ultimately the update time. Interestingly, the algorithmic approach that we develop for dynamizing our data structure can also be applied, at a high level, to Fiat and Naor's original construction -- this is presented in Section \ref{sec:fiat-naor-dynamic}.

\paragraph{An application to dynamic unordered graphs. }In addition to our results on function inversion, we also apply our results to the problem of storing a dynamic $(1 + \epsilon)$-compact unordered graph. Suppose we are given an unordered graph $G$ with labeled vertices $[n]$ and with $m$ edges, and we wish to store it in a data structure using close-to-optimal space, while supporting basic graph operations such as edge insertion/deletion, neighborhood queries, etc. A simple adjacency-list representation wastes (at least) a factor-of-two on space, because it stores each edge $\{i, j\}$ \emph{twice}, once from $i$ to $j$, and once from $j$ to $i$. This raises a natural question: Can one hope to get within a factor of $(1 + \epsilon)$ of the information-theoretically optimal space bound for the graph, while still supporting basic graph operations in close-to-constant time? 

The state of the art for this problem, due to Farzan and Munro \cite{FarzanMunro08} in 2008, is a \emph{static} data structure that supports $O(\epsilon^{-1})$-time neighborhood and adjacency queries. It has remained an open question whether a \emph{dynamic} solution to the same problem may exist. 

In Section \ref{sec:unordered-graphs}, we present a dynamic data structure that supports $O(\epsilon^{-1})$ time for every operation (neighborhood and adjacency queries, edge insertions, and edge deletions), while using space within a factor of $(1 + \epsilon)$ of information-theoretically optimal (and provided that $m = n^{2 - \Omega(1)}$). This result is summarized in the following theorem, which is proven in Section \ref{sec:unordered-graphs}:

\begin{restatedtheorem}{\ref{thm:dynamic-unordered-graphs}}
\thmunorderedgraphsbody
\end{restatedtheorem}

At a high level, our approach to proving Theorem \ref{thm:dynamic-unordered-graphs} is to encode the graph $G$ as a function $f_G$ that can be used to follow edges $(i, j)$ with $i < j$; and to then use function inversion to follow edges in the other direction $j \rightarrow i$. We then show how to encode the function $f_G$ efficiently (using close to optimal space) and in a way that allows us to support both adjacency and neighborhood queries (this requires several interesting data-structural ideas). Our approach also results in a new and arguably simpler solution to the static version of the problem.

\section{Related Work}
Hellman \cite{Hellman80} initiated the study of function inversion in $1980$. The analysis, which crucially assumes $f: [N] \to [N]$ is a random function, demonstrates a $O(N \log N / t)$-space solution with query time $\tilde{O}(t^2)$. 

In a breakthrough result of $1991$, Fiat and Naor presented a $\tilde{O}(t^3)$ query-time solution for arbitrary function inversion \cite{FiatNa91, FiatNa00}. Their work also extends to support query time $\tilde{O}(N t^2 q(f))$, where $q$ is the probability that two independent random elements of $[N]$ collide under $f$. As a corollary, they provide a query time of $\tilde{O}(t^2)$ for functions with maximum in-degree $\polylog(N)$ (a property that a random function satisfies with high probability in $N$).

Much more recently, Golovnev, Guo, Peters, and Stephens-Davidowitz \cite{golovnevGuPe23} presented a (non-uniform) solution that achieves a nontrivial time bound even for $t \in [N^{1/3}, \sqrt{N}]$ (or, equivalently, $S \in [\sqrt{N}, N^{2/3}]$). In this parameter regime, their result achieves inverse query time $\tilde{O}(t \sqrt{N})$. It remains open whether a nontrivial query time (i.e., with query time $o(N)$) is possible for $S \ll \sqrt{N}$. 

Whereas the original $\tilde{O}(t^3)$ bound by Fiat and Naor hides $\polylog N$ terms, a more recent version of the construction, by Bibbens, Borevitz, and McCauley \cite{BibbensBorevitzMcCauley2026TextIndexingMismatches}, reduces this to a $\polylog t$ term, thereby extending the original construction to offer nontrivial guarantees even for small values of $t$. 


On the lower bound side, the state of the art remains a paper by Yao \cite{Yao90} from $1990$. Yao proved that any solution using space $S = N \log N / t$ bits must incur expected query time $\Omega(t)$. Yao's result applies not only to function inversion, but also the special case of inverting a permutation $\pi: [N] \to [N]$. See also Theorem 14 of \cite{impagliazzo11} for an explicit proof of Yao's lower bound. 

Barkan, Biham, and Shamir \cite{barkan2006rigorous} show that if one restricts to a specific class of data structures (intuitively the class that naturally generalizes the Fiat and Naor solution), then there is a query-time lower bound of $\Omega(t^2 / \log N)$. Corrigan-Gibbs and Kogan \cite{corrigan2019function} show that, in general, any improvement to Yao's $\Omega(t)$ lower bound (for all data structures)
would imply new lower bounds for depth-two circuits with arbitrary gates.

A natural relaxation of function inversion is the problem of building a \emph{single} data structure that answers \emph{both} forward and inverse queries on $f$ (we will call this the \emph{function bi-evaluation problem}). Munro, Raman, Raman, and Rao \cite{munro2012succinct} give a solution using space $(1 + \epsilon) N \log N$ bits while supporting both function-evaluation and function inversion in time $O(\epsilon^{-1})$.\footnote{In fact, their solution supports evaluation of arbitrary powers $f^k(x)$, $k \in \mathbb{Z}$ in $O(\epsilon^{-1})$ time.} This result is achieved via a reduction to storing permutations, which can then be solved using a simple back-pointer construction due to Hellman \cite{Hellman80}. An immediate corollary of our results is a dynamic version of this result (allowing updates to $f$). 

A natural relaxation of function inversion is to build a \emph{single}
data structure that supports both forward evaluation and inverse-image
queries on a static function $f:[N]\to[N]$; we call this the
\emph{function bi-evaluation problem}.  Munro, Raman, Raman, and Rao
\cite{munro2012succinct} give a solution using
$(1+1/t)N\lg N+O(1)$ bits (for $t \in O(\log n / \log \log n)$
which computes positive powers $f^k(x)$ in time $O(t)$ and returns the set of
negative powers $f^{-k}(x)=\{y : f^k(y)=x\}$ 
in time $O(1+|f^{-k}(x)| \cdot t)$ for any integer $k$.
Their construction proceeds by reducing the problem to two other 
static problems: (1) the problem of encoding a permutation and its inverse,
for which Hellman's original construction suffices, and
(2) the problem of efficiently answering level-ancestor queries in a succinct
static rooted tree. It remains an open question whether a dynamic variant
of Munro et al.'s result exists \cite{munro2012succinct}; one corollary
of this paper is that, if one wishes only to perform $f$ and $f^{-1}$ queries
(not $f^k$ for arbitrary $k$), then a dynamic solution is possible.

We remark that, for algorithmic applications (even static ones), function inversion can, in some cases, be much more powerful than function bi-evaluation on its own. Consider, for example, the graph application in Section \ref{sec:dynamic-unordered-graphs}, which encodes a graph $G$ in a function $f_G$. Directly encoding $f_G$ using function bi-evaluation, even in the static setting, would result in an encoding that (1) uses $\approx (m + n) \log (m + n)$ bits, which is a substantial constant factor larger than the optimal $\log \binom{\binom{n}{2}}{m}$ bits (the issue is that the $f_G$ is not actually an \emph{arbitrary} function); and (2) fails to answer either adjacency queries or neighborhood queries (both of which, in our solution, exploit features of our specific data structure for forward-encoding $f_G$). For applications such as this one, the fact that function inversion can be used with an arbitrary (potentially problem-specific) encoding of the forward function $f$ (optimizing both space and functionality over a naive encoding) is critical. 

Whereas early work on function inversion was motivated primarily by cryptographic applications, there has been a significant line of recent works showing how to apply it as an \emph{algorithmic tool} in a variety of other settings \cite{
    KopelowitzPorat2019Strong3SUMIndexing,
    GolovnevEtAl2020DataStructuresMeetCryptography,
    AronovEtAl2023CollinearityIndexing,
    BilleEtAl2024GappedStringIndexing,
    AronovEtAl2024ImplicitSets,
    McCauley2024ANNFunctionInversion,
    MazorPass2024Perebor,
    HiraharaIlangoWilliams2024Compression,
    DinurGolovnev2026Improved3SUMIndexing,
    KirkpatrickEtAl2026Preprocessed3SUMUnknownUniverses,
    BibbensBorevitzMcCauley2026TextIndexingMismatches,
    CharalampopoulosEtAl2026SuffixRandomAccess
}. This phenomenon is one of the primary motivations of the current paper. We are optimistic that, by extending function inversion to be both dynamic and, in the low-time regime, optimal, we will enable additional algorithmic applications in the future.

\section{Preliminaries}

In the function-inversion problem, one is given constant-time oracle access to a function $f:[N] \rightarrow[N]$. One must then construct a data structure which (with help from the $f$-oracle) answers $f$-inverse queries, where for any $y \in [N]$, $\mathsf{Inverse}(y)$ returns an element in $f^{-1}(y)$, if $f^{-1}(y) \neq \emptyset$, or $\bot$ if $f^{-1}(y) = \emptyset$. The space of this inversion data structure is typically denoted by $S = (N \log N) / t$ bits, and the goal is to achieve a small query time $T$ as a function of $t$. In cases where $T = \Theta(t)$, we will often discard the variable $t$ and use use $T$ only (so the space is $S = \Theta((N \log N) / T)$).

In this paper, we will also be interested in the dynamic version of the problem, in which we must support \emph{updates to $f$}. For $x, y \in [N]$, an update $\textsf{Update}(x, y)$ updates $f(x)$ to be $y$. In the dynamic function-inversion problem, the $f$-oracle updates automatically (always serving as an oracle for the current version of $f$), and the goal is to construct a function-inversion data structure that supports both efficient inverse queries and efficient updates. 

In developing dynamic function-inversion data structures, a crucial intermediate step will be to extend function inversion to recover not just a single inverse, but all of the inverses of an element $y$. Specifically, we will be interested in supporting \emph{iterators} which can be used as follows: an iterator for an element $y \in [N]$ can be initialized by calling a function $\mathsf{Iterate}(y)$ which returns an iterator $q$; one can then loop through the inverses $x \in f^{-1}(y)$ one by one with calls to $\mathsf{Next}(q)$, which either returns the next inverse (the order is up to the data structure) or returns $\bot$ if there are no more inverses. Whenever an update occurs, all existing iterators become invalid.

When analyzing our data structures, we will often want to support high-probability guarantees. Specifically, we will use the term \defn{with high probability in $N$} to mean with probability at least $1 - 1 / \poly(N)$ for an arbitrarily large polynomial $\poly(N)$ of our choice. 

Finally, several of our constructions will make use of the following retrieval data structure due to Demaine, Meyer auf der Heide, Pagh, and P\u{a}tra\c{s}cu~\cite{demaine2006dictionariis}:

\begin{proposition}
There is a retrieval data structure that, given a set $S \subseteq [N]$ of at most $M$ elements, and given a function $f: S \rightarrow [K]$ allows one to evaluate $f(x)$ for any $x \in S$ in $O(1)$ time, while using $O(M \log \log N + N \log K)$ bits of space, and supporting updates to both $S$ (via insertion/deletion) and $f$ (when an element inserted/deleted to $S$) in $O(1)$ time, with high probability in $M$. For elements $x \not\in S$, the data structure has no guarantees on what value it returns. \label{prop:retrieval}
\end{proposition}

Note that, when applying Proposition \ref{prop:retrieval}, we will typically assume implicitly that $M \ge N^{\Omega(1)}$ (that way the probabilistic guarantee is with high probability in $N$). If this is not true, then it can be made so by artificially increasing $M$ to $N^{\Omega(1)}$ (and in our applications this will have negligible impact on overall space).

\section{Proving Theorem \ref{thm:main}}\label{sec:main}

To prove Theorem \ref{thm:main}, we present a sequence of constructions that each achieve stronger and stronger guarantees. This lets us isolate out each of the algorithmic ingredients that are needed for the full construction.

\subsection{Achieving Query Time $O(T)$ and Space $O(N \log N / T)$ Bits}

We begin in this section by presenting a simple function-inversion data structure that, in the regime of $T = O(\log N / \log \log N)$, achieves the optimal query-time/space tradeoff, supporting query time $O(T)$ and space $O(N \log N / T)$ bits. As we will discuss, this initial data structure is still far from achieving Theorem \ref{thm:main}: Its construction time is $O(N \polylog N)$, it does not dynamize, it offers only a single inverse for each invertible element, and it assumes free randomness. Nonetheless, it will serve as a good starting point for the development of the full construction.

\begin{theorem}
Let $f: [N] \to [N]$ be a function which we have $O(1)$-time oracle access to, and assume free randomness. For $T = O(\log N / \log \log N)$, we can construct in $O(N \polylog N)$ time a data structure that supports $O(T)$-time $f$-inverse queries, and that uses space $O(N \log N / T)$ bits, where the construction succeeds with high probability in $N$.
\label{thm:startingpoint}
\end{theorem}

The construction (both in this section and in later section) will make use of a set of positive constants $\alpha_1, \alpha_2, \alpha_3$, where $\alpha_1$ is a sufficiently large positive constant, $\alpha_2$ is a sufficiently large constant even compared to $\alpha_1$, and $\alpha_3$ is a sufficiently large constant even compared to $\alpha_1, \alpha_2$.

Our data structure will utilize $B = (\log N)^{\alpha_1}$ uniformly random hash functions $g_1, g_2, \ldots, g_B: [N] \to [N]$. We also store an array $\mathcal{D}$, which contains for each $y \in [N]$, a choice $j \in [B]$ of hash function to apply to $y$. This requires $O(N \log B) = O(N \log\log N)$ bits. This choice of $j$ will be determined during the construction process, and the resulting $g_j(y)$ will be denoted simply as $g(y)$. We will further define a function $h = (g \circ f)$ so that, for a preimage $x$, we have $h(x) = g(f(x)) = g_{\mathcal{D}[f(x)]}(f(x))$.

\paragraph{The basic data structure.}
The construction will construct a set of disjoint \defn{chains} of length $T$, where the members of a chain are preimages. Each chain $i$ starts at some preimage $s_i \in [N]$, ends at some preimage $e_i \in [N]$, and consists of preimages $s_i, h(s_i), h^2(s_i), \ldots, h^{T - 1}(s_i) = e_i$. We store a lookup table $\mathcal{C}$ mapping, for each chain $i$, the value $f(e_i)$ to the corresponding starting preimage $s_i$.

Say that an element $y$ is \defn{covered} by a chain if it has at least one preimage that is a member of the chain. Our construction will ensure that every covered element is covered just once, by a single one of its preimages in a single chain. We will design the chains in such a way that all but $O(N / \log N)$ invertible elements are covered. The uncovered invertible elements will have an inverse stored directly in a look-up table $\mathcal{H}$ of size $O(N \log N / \log N)$ bits.

Assuming a data structure of this form, inverse queries can be answered as follows, in $O(T)$ time (see, also, Algorithm \ref{alg:query}). To find an inverse for an element $y \in [N]$ (or determine that $y$ is not invertible), we proceed as follows. First we check if the inverse is stored in $\mathcal{H}$. Barring this, we proceed as follows. We compute $y, f(g(y)), f(g(f(g(y)))), \ldots$, for $T$ total values in the sequence (with the help of $\mathcal{D}$ to determine which $g_i$ to use in each step), and if any of these appear as a key in $\mathcal{C}$, then we take the first such key and use $\mathcal{C}$ to recover the corresponding starting preimage $s_i$. We then replay the chain $s_i, h(s_i), h^2(s_i), \ldots, h^{T - 1}(s_i)$, checking whether $f$ of any chain member is equal to $y$. If $y$ is covered by a chain, then (since no element is covered by more than one chain member), this replay will encounter the unique chain member $w$ satisfying $f(w)=y$, and we return $w$ as an inverse for $y$. Finally, if we did not find any key in $\mathcal{C}$ or $y$ was not covered by the resulting chain, then we can conclude that $y$ is not covered by any chain. Since $y$'s inverse was also not stored in $\mathcal{H}$, this means that $y$ is not invertible.

\begin{algorithm}[t]
    \caption{\(\mathsf{Query}(y)\)}
    \label{alg:query}
    \begin{algorithmic}[1]
    \Function{Replay}{$y,s$}
        \Comment{\(s\) is the first preimage in the chain}
        \State \(w \gets s\)
        \State \(z \gets f(w)\)
    
        \For{\(r=0,\ldots,T-1\)}
            \If{\(z=y\)}
                \State \Return \(w\)
            \EndIf
            \State \(k \gets \mathcal{D}[z]\)
            \State \(w \gets g_k(z)\)
            \Comment{Now \(w\) is the next chain member}
            \State \(z \gets f(w)\)
        \EndFor
    
        \State \Return \(\bot\)
    \EndFunction
    
    \Statex
    
    \If{\(\mathcal{H}\) contains key \(y\)}
        \State \Return \(\mathcal{H}[y]\)
    \EndIf
    
    \State \(z \gets y\)
    
    \For{\(r=0,\ldots,T-1\)}
        \If{\(\mathcal{C}\) contains key \(z\)}
            \State \(s \gets \mathcal{C}[z]\)
            \Comment{\(s\) is the first preimage in the chain}
            \State \Return \Call{Replay}{$y,s$}
        \EndIf
        \State \(k \gets \mathcal{D}[z]\)
        \State \(w \gets g_k(z)\)
        \State \(z \gets f(w)\)
    \EndFor
    
    \State \Return \(\bot\)
    \end{algorithmic}
    \end{algorithm}

\paragraph{Constructing the chains.} We can construct the chains as follows. We construct the chains one by one, until there are fewer than $N/\log N$ invertible elements remaining that are not covered by chains. As we proceed, we keep track for each element if it is covered by a chain. To construct a new chain, we first find an arbitrary preimage $x$ such that $f(x)$ is not yet covered by any chain (we do this by repeatedly randomly sampling $x$s until we find one for which $f(x)$ is not yet covered). 

We then build the chain $s_i = y$, $h(s_i), h^2(s_i), \ldots, h^{T - 1}(s_i) = e_i$, as follows, in such a way that the elements covered by the chain are disjoint from the elements covered by all prior chains. Given a preimage $z$ in the chain, we determine the next preimage $h(z) = g(f(z))$ by determining which $g_k$ to apply to $f(z)$ (and storing $k$ in $\mathcal{D}[f(z)]$). To select which $g_k$ to use, we try each of the $g_k$s, $k = 1, 2, \ldots, B$, in order, until we find one such that $f(g_k(f(z)))$ is not yet covered by any chain (if no $g_k$ exists, we declare failure for the entire construction). Note that, since $f(z)$ has never previously been covered by any prior chain, there is no risk that some prior chain has already decided which $g_k$ to use for $f(z)$.

Finally, having constructed the chain, we update \(\mathcal{C}\) to map
\(f(e_i)\) to the starting preimage \(s_i\). We continue this process until there are fewer than $N/\log N$ invertible elements remaining that are not covered by chains. (We may keep track of the number of such elements by calculating at the start the number of invertible elements $R$, and subtracting the number of such elements covered by chains.) After constructing the chains, we spend $O(N)$ time building a lookup table $\mathcal{H}$ to answer the inverse queries for all invertible elements not yet covered by chains. By construction, this table uses $O(N \log N / \log N)$ bits.

Since this construction will serve as the starting point for all of our constructions in later sections, we include pseudocode for it in Algorithm \ref{alg:construction}.

\begin{algorithm}[H]
    \caption{\(\mathsf{Construct}(f,g_1,\ldots,g_B,T)\)}
    \label{alg:construction}
    \begin{algorithmic}[1]
    \State Initialize \(\mathcal{D}[y]\gets1\) for all \(y\in[N]\) \Comment{This is the default value used when $\mathcal{D}[y]$ is not yet set}
    \State Initialize empty lookup tables \(\mathcal{C}\) and \(\mathcal{H}\)
    \State Initialize \(\mathsf{covered}[y]\gets\mathsf{false}\) for all \(y\in[N]\)
    \State Compute \(R \gets |\{y\in[N] : f^{-1}(y) \neq \emptyset\}|\)
    \Comment{Number of invertible elements}
    
    \Statex
    
    \Function{FindStart}{}
        \Repeat
            \State Sample \(x\in[N]\) uniformly at random
            \State \(y\gets f(x)\)
        \Until{\(\mathsf{covered}[y]=\mathsf{false}\)}
        \State \Return \(x\)
        \Comment{\(f(x)\) is not yet covered}
    \EndFunction
    
    \Statex
    
    \Function{ChooseExtension}{$z$}
        \For{\(k=1,\ldots,B\)}
            \State \(x\gets g_k(f(z))\)
            \State \(y'\gets f(x)\)
    
            \If{\(\mathsf{covered}[y']=\mathsf{false}\)}
                \State \Return \((k,x,y')\)
                \Comment{\(x\in f^{-1}(y')\)}
            \EndIf
        \EndFor
    
        \State \Return \((\bot,\bot,\bot)\)
    \EndFunction
    
    \Statex
    
    \While{\(R\ge N/\log N\)}
        \State \(s\gets \Call{FindStart}{}\)
        \State \(\mathsf{covered}[f(s)] \gets \mathsf{true}\)
        \State $R \gets R - 1$
        \State \(z\gets s\)
    
        \For{\(r=1,\ldots,T-1\)}
            \State \((k,w,y')\gets \Call{ChooseExtension}{z}\)
    
            \If{\(k=\bot\)}
                \State \textbf{abort}
            \EndIf
    
            \State \(\mathcal{D}[f(z)]\gets k\)
            \State \(z\gets w\)
            \State $\mathsf{covered}[y'] \gets \mathsf{true}$
            \State $R \gets R - 1$
        \EndFor
    
        \State \(e\gets z\)
        \State \(\mathcal{C}[f(e)]\gets s\)
    \EndWhile
    
    \For{\(x=1,\ldots,N\)}
        \If{\(\mathsf{covered}[f(x)]=\mathsf{false}\)}
            \State \(\mathcal{H}[f(x)]\gets x\)
            \State $\mathsf{covered}[f(x)] \gets \mathsf{true}$
        \EndIf
    \EndFor
    
    \State \Return \((\mathcal{D},\mathcal{C},\mathcal{H})\)
    \end{algorithmic}
    \end{algorithm}

\paragraph{Analysis. } By construction, the data structure $\mathcal{D}$ uses $O(N \log B) = O(N \log\log N)$ bits. The lookup table $\mathcal{C}$ uses $O(\log N)$ bits per chain. Since the chains are disjoint and length $T$, the lookup table $\mathcal{C}$ uses total space $O(N (\log N)/T)$ bits. The lookup table $\mathcal{H}$ stores at most $O(N / \log N)$ inverses, and therefore uses space $O(N)$ bits. Therefore, the total space of the data structure is $O(N \log N / T)$ bits.

The query time is $O(T)$ in the worst case. So it remains to analyze both the success probability and the running time of the construction.
\begin{lemma}
With high probability in $N$, the construction succeeds.
\label{lem:construction-success}
\end{lemma}
\begin{proof}
The only way for the construction to fail is if, when trying to determine which $g_k$ to use for a preimage $z$, we exhaust all $B$ hash functions without finding one such that $f(g_k(f(z)))$ is uncovered. 

By construction, when we started building the chain, there were at least $N / \log N$ uncovered invertible elements. The chain consumes at most $T = O(\log N / \log \log N)$ of these, so there remain at least, say, $N / (2 \log N)$ uncovered invertible elements. Therefore, each $g_k$ independently has probability at least $(N / (2 \log N)) / N = 1 / (2 \log N)$ of satisfying the condition that $f(g_k(f(z)))$ is uncovered. The probability of no $g_k$ satisfying this property is therefore at most $(1 - 1 / (2 \log N))^B \le 1 / \poly(N)$.
\end{proof}

Having established that the construction succeeds, we also bound its running time as follows:
\begin{lemma}
With high probability in $N$, the construction takes time $O(N \polylog N)$.
\label{lem:NpolylogN}
\end{lemma}
\begin{proof}
The time to determine which $g_k$ to use for a preimage $z$ is $O(B) = O((\log N)^{\alpha_1})$. Therefore, the total time spent choosing $g_k$s for all chain members is $O(N \polylog N)$. Additionally, each chain must \emph{find} a starting preimage $x$ by sampling random $x$s until we find one for which $f(x)$ is not yet covered. Since there are at least $N / \log N$ such $x$s remaining, this sampling step takes $\polylog N$ time with high probability in $N$. Summing over the chains, this contributes another $O(N \polylog N)$ time. All other steps in the construction are $O(N)$ time, so the total construction time is $O(N \polylog N)$.
\end{proof}

Combining these guarantees gives Theorem \ref{thm:startingpoint}.

\subsection{Achieving Linear Construction Time}\label{sec:linear}

In this section, we show how to modify the data structure above to achieve $O(N)$ time for preprocessing. 

\begin{theorem}
Let $f: [N] \to [N]$ be a function which we have $O(1)$-time oracle access to, and assume free randomness. For $T = O(\log N / \log \log N)$, we can construct in linear time a data structure that supports $O(T)$ $f$-inverse queries, and that uses space $O(N \log N / T)$ bits, where the construction succeeds with high probability in $N$.
\label{thm:linear}
\end{theorem}

The main challenge is to select, for a given preimage $x$ in a chain, which $g_i$, $i \in [B]$, to use in $h(x) := g_i(f(x))$. This choice must be made in $O(1)$ (amortized) time per element, and must be made in a way that avoids any collisions between different chains. 

\paragraph{Step 1: Imposing structure on the $g_i$s. } The main idea in this section is to redefine the $g_i$s as follows. Let $\pi$ be a uniformly random permutation of $[N]$, let $r: [N] \to [N/B]$ be a uniformly random hash function, and assume we can access both $\pi$ and $r$ in $O(1)$ time. (We will show in Section \ref{sec:random} how to simulate $\pi$ and $r$ using $o(N)$ bits, but for now we will assume they are fully random.) Using $\pi$, we define \defn{bins} of size $B$, where bin $i$ consists of $\pi((i-1) B + 1), \pi((i-1) B + 2), \ldots, \pi(i B)$ for all $i = 1, 2, \ldots, N/B$. As a shorthand, we refer to element $\pi((i-1) B + j)$ as the $j$-th element of bin $i$. Then, for a given input $u$, we define $g_i(u)$ to be the $i$-th element of bin $r(u)$, given by $g_i(u) = \pi((r(u) - 1) B + i)$.

\paragraph{Step 2: Targeting specific inverses. } The second change to the data structure is a bit more subtle, as its role only becomes clear in the analysis. Before building the chains, we select for each invertible element $y$ a \emph{unique} \defn{target inverse} $\hat{f}^{-1}(y) \in f^{-1}(y)$. We do this in a way that is independent of both $r$ and $g$, and that uses an auxiliary data structure of size $O(N)$ bits: we loop through the elements $x \in [N]$, check for each $x$ if $y = f(x)$ has already found an inverse, and if it has not, we set $\hat{f}^{-1}(y) = x$. This can be done using two additional arrays of size $N$ to keep track of both which elements $y$ have found target inverses, and which elements $x$ have been selected as target inverses for some $y$.

\paragraph{Step 3: Building the chains. } Finally, we can now describe the updated procedure for building a chain. Each bin $i$ stores an $O(\log \log N)$ bit counter $c_i \in [B]$, initialized to $1$. We also store a bit map indicating which elements $i \in [N]$ are target inverses that have not yet been included in a chain (call these \defn{free targets}). Given a preimage $x$ in the chain, to determine which $g_i$ to use, we proceed as follows: check whether the $c_{r(f(x))}$-th element of the bin is a free target; if it is not, we increment $c_{r(f(x))}$, and try again; and so on, until we either find a free target or exhaust the bin (meaning that $c_{r(f(x))}$ is equal to the number of elements in the bin). If we exhaust the bin, we declare failure (we will prove that this happens with very low probability). Otherwise, we set $g_i(f(x))$ to use $i = c_{r(f(x))}$, meaning that $g_i(f(x))$ is a free target $x'$, and we increment $c_{r(f(x))}$ one more time (again declaring failure if it becomes an invalid bin index); finally, we also update our bit map of free target inverses to mark $x'$ as no longer free, and we store the choice $i$ in $\mathcal{D}[f(x)]$. Having computed $x'$, we use $x'$ as the next member of the chain, and this new member covers $f(x')$.

Finally, there is the issue of identifying the \emph{first element} $s_i$ of a chain. In order to start a new chain, we select a random bin $b \in [N/B]$, and use the same process as above to find an entry $x$ in the bin that is a free target (as before, we increment the counter $c_b$ as we go, and declare failure if we exhaust the bin). We then set $s_i = x$ as the starting preimage of the chain, and we consider $x$ to no longer be a free target.

We stop building chains once there are $N / \log^3 N$ free targets remaining. Note that we could reasonably stop when there are $N / \log N$ left, as in the previous section, but going until we get to $N / \log^3 N$ will be helpful when we modify the data structure in the next section. For each such free target $x$, we store the fact that $x$ is an inverse of $f(x)$ explicitly in the look-up table $\mathcal{H}$ of size $O(N \log N / \log^3 N) = O(N / \log^2 N)$ bits.

\paragraph{Analysis. }By construction, the data structure uses $O(N \log \log N)$ bits for various data structures, along with $O(N \log N / T)$ bits to store the values $f(e_i)$ and the starts of chains. Since $T = O(\log N / \log \log N)$, this gives a total space of $O(N \log N / T)$ bits.

To reason about construction time, observe that the total time spent selecting $g_i$s is equal, up to constant factors, to the total number of times we increment the counters $c_j$. Since each counter loops through the elements in a bin, the counters can never sum to more than $O(N)$. Therefore, the total time spent selecting $g_i$s is $O(N)$. Since the other steps in the construction also take $O(N)$ time, the total construction time is $O(N)$.

Finally, the most interesting step in the analysis is to argue that, with high probability in $N$, the data structure does not declare failure when building the chains. Here, we will see that pre-selecting a target inverse for each invertible element $y$ in the domain allows us to avoid collisions between different chains.

\begin{lemma}
With high probability in $N$, the construction succeeds, meaning that every bin contains at least one free target at all points in time. 
\label{lem:success}
\end{lemma}
\begin{proof}
It suffices to consider a fixed bin $i$ and a fixed point in time $t$, and to show that with high probability in $N$, bin $i$ contains at least one free target at time $t$.

Let $X$ be the number of target inverses $x = \hat{f}^{-1}(y)$ that are placed in bin $i$ by the permutation $\pi$, at the start of the construction. Then, if $R$ is the number of elements that have at least one inverse, we have $\E[X] = RB / N \le B$. Moreover, $X$ is a sum of negatively associated iid 0-1 random variables, indicating for each target inverse whether it is in the bin. Therefore, we can apply a Chernoff bound to deduce that, with high probability in $N$, 
\begin{equation}
X \ge \E[X] - O(\sqrt{B \log N}) \ge RB/N - O(\sqrt{B \log N}).
\label{eq:Xch}
\end{equation}

Let $Y$ be the number of target inverses in bin $i$ that have been consumed (i.e., that are no longer free) by the time $t$. Each time we add a preimage to a chain (this includes the first element of the chain), we select a random bin and use a target inverse from that bin: for the first element this bin is chosen explicitly, and for each later element it is bin $r(f(x))$, where $x$ is the previous chain member. The total number of times we do this is at most $R - \Omega(N / \log^3 N)$ (since we stop building chains once there are $N / \log^3 N$ free targets remaining). Therefore, $Y$ is at most a sum $Y'$ of $R - \Omega(N / \log^3 N)$ iid 0-1 random variables, each of which is $1$ with probability $B/N$. The mean for $Y'$ is $(R - \Omega(N / \log^3 N)) / (N / B) = RB / N - \Omega(B  / \log^3 N)$. By a Chernoff bound, we may deduce that with high probability in $N$, 
\begin{equation}
Y \le Y' \le \E[Y'] + O(\sqrt{B \log N}) \le RB / N - \Omega(B  / \log^3 N) + O(\sqrt{B \log N}).
\label{eq:Ych}
\end{equation}

Combining \eqref{eq:Xch} and \eqref{eq:Ych}, we deduce that with high probability in $N$ that
$$X - Y \ge \Omega(B / \log^3 N) - O(\sqrt{B \log N}).$$
Since $B = (\log N)^{\alpha_1}$ where $\alpha_1$ is a sufficiently large positive constant, the above expression is positive, as desired. 
\end{proof}

Combining these guarantees gives Theorem \ref{thm:linear}.

\subsection{Recovering All Inverses for Non-Heavy-Hitters}

Our next major goal is to extend the function-inversion data structure to return not just a single inverse, but all of the inverses of each element. This, it turns out, will be the most important ingredient for allowing us to, later on, dynamize the data structure.

Formally, we will extend the data structure to support a function $\mathsf{Inverse}_f(y, i)$ with the following guarantee: if an element $y$ has $j$ inverses, then there exists some permutation $\phi: [j] \rightarrow f^{-1}(y)$ such that $\mathsf{Inverse}_f(y, i) = \phi(i)$, for $i \le j$, and $\mathsf{Inverse}_f(y, i) = \bot$ for $i > j$. This is to say that $\mathsf{Inverse}_f(y, \cdot)$ returns all of the inverses of $y$ in some order determined by the data structure.

To begin, in this section, we will make a simplifying assumption: that there are no so-called \defn{heavy hitters}, which are elements $y$ with $|f^{-1}(y)| \ge \log^2 N$. We will then show how to handle heavy hitters separately in the next section.

The goal of this section, then, is to prove the following theorem:
\begin{theorem}
    Let $f: [N] \to [N]$ be a function which we have $O(1)$-time oracle access to, and which has no heavy hitters, and assume free randomness. For $T = O(\log N / \log \log N)$, we can construct in linear time a data structure that supports $O(T)$-time $\mathsf{Inverse}_f(y, i)$ queries, and that uses space $O(N \log N / T)$ bits, where the construction succeeds with high probability in $N$.
\label{thm:completeinversion-nohitters}
\end{theorem}

To begin our construction, first break the elements $x \in [N]$ into $\log^2 N$ groups $P_1, \ldots, P_{\log^2 N}$, where group $j$ contains one inverse for each element $y$ that has at least $j$ inverses. This can straightforwardly be done in $O(N)$ time using $O(N \log \log N)$ bits of space.\footnote{Simply evaluate $f$ on every element $x$, while keeping track for each element $y$ of the number of inverses $\gamma_y$ found so far; if, when $f(x)$ is computed, this increases $\gamma_{f(x)}$ to some number $j$, then $x$ goes to group $j$. Note that the counters $\{\gamma_y\}$ are $O(\log(\log^2 N)) = O(\log \log N)$ bits each.} 

A natural approach to returning all inverses is to construct $\log^2 N$ different data structures $D_1, D_2, \ldots, D_{\log^2 N}$, where data structure $D_k$ recovers inverses from group $k$ (or nothing, if an element $y$ has no group-$k$ inverse). This can be done by repeating the construction in Subsection \ref{sec:linear} $\log^2 N$ times, and modifying the $k$-th iteration to use as its ``target inverses'' the elements in group $P_k$; as before, the construction of $D_k$ terminates when there are at most $N / \log^3 N$ free targets remaining, and then uses an explicit lookup table $\mathcal{H}$ (now $\mathcal{H}_k$) to store for each remaining free target $x$, a mapping from $f(x)$ to $x$. By the same analysis as in Lemma \ref{lem:success}, each of these $D_k$ constructions will succeed with high probability in $N$.

This naive solution succeeds in supporting the $\mathsf{Inverse}_f(y, k)$ interface with $O(t)$ query time. But, as stated so far, it is both space-inefficient and problematically slow to construct. We address each of these issues separately below. 

\paragraph{Bringing the space down to $O(N \log N / T)$ bits. } Our first task is to modify the construction so that the total space used by all of $D_1, D_2, \ldots, D_{\log^2 N}$ is $O(N \log N / T)$ bits.

First notice that, if each $D_k$ uses a separate $\mathcal{D}$ array (denoted by $\mathcal{D}_k$), then the total space will blow up to $O(N \log^2 N)$ bits, which is unacceptable. To rectify this, we replace the $\mathcal{D}$ array with a retrieval data structure $\mathcal{R}_k$ (Proposition \ref{prop:retrieval}) mapping each element $y$ covered by any chain of $D_k$ to the $O(\log \log N)$-bit value $\mathcal{D}_k[y]$. The data structure $\mathcal{R}_k$ uses $O(N_k \log \log N)$ bits where $N_k$ is the number of elements covered by chains in $D_k$. By construction, $\sum_k N_k \le N$, so the $\mathcal{R}_k$ data structures collectively use $O(N \log \log N)$ bits.

With this modification in mind, we can bound the total space used by the $\log^2 N$ data structures as follows:
\begin{itemize}
\item $O(N \log \log N)$ bits to store the $\mathcal{R}_k$s;
\item $O(N / \log^2 N)$ bits to store the lookup table $\mathcal{H}$, for each $D_k$; this sums to $O(N / \log N)$ bits. 
\item $O(|P_k| (\log N) / T)$ bits to store the dictionary $\mathcal{C}$, for each $D_k$; since $\sum_k |P_k| = N$, this sums to $O(N \log N / T)$ bits.
\end{itemize}
The total space is therefore $O(N \log \log N) + O(N (\log N) / T) = O(N (\log N) / T)$ bits, overall, which is within the space budget of Theorem \ref{thm:completeinversion-nohitters}.

\paragraph{Bringing the construction time down to $O(N)$. } Our second task is to modify the construction so that the total construction time is $O(N)$.

Here, the only issue is the task of determining which $g_i$ to use when defining $h(x) = g_i(f(x))$ for some element $x$ in a chain. To do this, we must identify some target inverse in bin $r(f(x))$. In the construction from Subsection \ref{sec:linear}, this is done by examining the $c_{r(f(x))}$ counter, and repeatedly incrementing it until we get to a value such that the $c_{r(f(x))}$-th element in bin $r(f(x))$ is a free target. This approach led to an $O(N)$ construction time, in Subsection \ref{sec:linear}, but would now give us an $O(N \log^2 N)$ construction time when applied to each of $D_1, D_2, \ldots, D_{\log^2 N}$.

To fix this issue, we eliminate the counters $c_{r(f(x))}$, and use a more efficient data-structure to store the remaining free targets in each bin. For each bin $i$, we construct a data structure $\mathcal{L}_i$ of $O(B \log \log N + \log^2 N \log \log N) = O(B \log \log N)$ bits that keeps track of, for each group $P_k$, the (remaining) free targets in bin $i$ for that group. Define for each group $k$ the set $T_{k, i}$ to be the set of elements $q \in [B]$ such that the $q$-th element of the bin $i$ is a free target in group $k$. The data structure $\mathcal{L}_i$ consists of an array of $B$ $(\log B)$-bit numbers, and of $\log^2 N$ linked lists, where the $k$-th linked list visits, within the $B$-sized array, exactly the positions in the set $T_{k, i}$. This is to say that the head of the $k$-th linked list is a number $t_{k, 1} \in T_{k, i}$, that the $t_{k, 1}$-th element of the $B$-sized array is another number $t_{k, 2} \in T_{k, i}$, and so on. Critically, because each linked-list `pointer' is a number in $[B]$, the total space used by the data structure $\mathcal{L}_i$ for bin $i$ is $O(B \log \log N + \log^2 N \log \log N) = O(B \log \log N)$ bits. Summing over the $N / B$ bins, this amounts to $O(N \log \log N)$ bits overall. 

Note that the initial states of the data structures $\{\mathcal{L}_i \mid i \in [N/B]\}$ can be constructed in $O(N)$ time since, whenever we discover that the $j$-th element of some bin $i$ is a group-$k$ inverse, we can add it the $k$-th linked list of $\mathcal{L}_i$ in $O(1)$ time. Each data structure $\mathcal{L}_i$ can then be used to find free targets in $O(1)$ time by simply following the appropriate linked list (and then removing that element from the linked list, since it will no longer be free in the future). 

With this modification in place, we can now bound the total construction time as follows. The total time spent on tasks besides constructing chains is straightforwardly $O(N)$. The total time spent building chains is also $O(N)$, since the time to build a given chain (whose length is $T$) is now $O(T)$, and there are $O(N / T)$ total chains to build, across all of the data structures $D_1, D_2, \ldots, D_{\log^2 N}$. Therefore, the total construction time is $O(N)$. 

This completes the proof of Theorem \ref{thm:completeinversion-nohitters}.

\subsection{Recovering All Inverses for Heavy Hitters}\label{sec:heavy-hitters-to-non-heavy-hitters}

Next, we extend Theorem \ref{thm:completeinversion-nohitters} to handle heavy hitters as well.

\begin{theorem}
    Let $f: [N] \to [N]$ be a function which we have $O(1)$-time oracle access to, and assume free randomness. For $T = O(\log N / \log \log N)$, we can construct in linear time a data structure that supports $O(T)$-time $\mathsf{Inverse}_f(y, i)$ queries, and that uses space $O(N \log N / T)$ bits, where the construction succeeds with high probability in $N$.
\label{thm:completeinversion-withhitters}
\end{theorem}

Note that the notion of a ``heavy hitter'' exists in essentially all known function-inversion data structures. Typically, to handle heavy hitters, one constructs a separate data structure that stores, for each heavy hitter $y$, a single inverse $x \in f^{-1}(y)$. This approach limits us to a single inverse (or, at least, a small subset of the inverses) for each heavy hitter. To recover all inverses, we will need a different approach.

\paragraph{Turning each heavy hitter into a set of non-heavy hitters. } For each heavy hitter $y$, let 
\begin{equation}
    m(y) = \left\lceil \frac{2|f^{-1}(y)|}{\log^2 N} \right\rceil,
    \label{eq:mdefn}
\end{equation}
and define $y_1, \ldots, y_{m(y)}$ to be the integers in the interval
$$\left( N + \sum_{\text{heavy hitter }y' < y} m(y'),\; N + \sum_{\text{heavy hitter }y' \le y} m(y') \right].$$
Let $\phi:[N] \rightarrow [0, 1]$ be a random hash function, let $N' = \sum_{\text{heavy hitter }y} m(y)$, and define a new function $f':[N'] \rightarrow [N']$ by 
$$
f'(x) = \begin{cases}
0 & \text{if } x > N, \\
f(x) & \text{if } x \in [N] \text{ and } f(x) \text{ is not a heavy hitter}, \\
f(x)_{\lceil m(f(x)) \cdot \phi(x) \rceil} & \text{if } x \in [N] \text{ and } f(x) \text{ is a heavy hitter}.
\end{cases}
$$
The basic idea behind $f'$ is that, for each heavy hitter $y$, it reroutes each element $x \in f^{-1}(y)$ to a random element among $y_1, \ldots, y_{m(y)}$. By distributing the preimage of $y$ among $m(y)$ new elements, we can eliminate heavy hitters.

\begin{lemma}
With high probability in $N$, $f'$ has no heavy hitters. Moreover, for each heavy hitter $y$ and each $i \in [m(y)]$, we have $|f'^{-1}(y_i)| = \Theta(\log^2 N)$.
\label{lem:phi}
\end{lemma}
\begin{proof}
Notice that 
$$|f'^{-1}(y_i)| = \sum_{x \in f^{-1}(y)} \mathbbm{1}\!\left[\left\lceil m(y) \cdot \phi(x) \right\rceil = i\right]$$
is a sum of $|f^{-1}(y)|$ iid 0-1 random variables with total mean $|f^{-1}(y)| / m(y) \in [0.25 \log^2 N, 0.5 \log^2 N]$. By a Chernoff bound, we have that with high probability in $N$ that $|f'^{-1}(y_i)| \in [\Omega(\log^2 N), \log^2 N]$.  Union bounding over all heavy hitters $y$ and all $i \in [m(y)]$, the lemma follows. 
\end{proof}

Note that $N'$ is only slightly larger than $N$. By design, 
$$N'- N \le \sum_{\text{heavy hitter }y} m(y) \le \sum_{\text{heavy hitter }y} |f^{-1}(y)| / \log^2 N = O(N / \log^2 N).$$
Thus we can afford to store $O((N' - N) \log N)$ (or even $O((N' - N) \log^2 N)$) bits of metadata in our data structure. With this in mind, it is straightforward to construct $f'$ in $O(N)$ time, and to store all of the data necessary to evaluate $f'$ (using $f$ as a black-box oracle) and using space $O(N \log \log N) = O(N \log N / T)$ bits. Notably, we can calculate who the heavy hitters are and what the $m(y)$ values are in $O(N)$ time using $N$ $O(\log \log N)$-bit counters; and then we can store for each heavy hitter $y$ and each $i \in [m(y)]$ a hash table mapping $(y, i)$ to $y_i$ (and vice-versa). 

After constructing $f'$, we directly apply Theorem \ref{thm:completeinversion-nohitters} to $f'$ to obtain a data structure that supports $O(T)$-time $\mathsf{Inverse}_{f'}(y, i)$ queries, and that uses space $O(N' \log N' / T) = O(N \log N / T)$ bits. This leaves us with only one remaining task: to support $\mathsf{Inverse}_f(y, i)$ queries in terms of $\mathsf{Inverse}_{f'}(y, i)$ queries. 

\paragraph{Converting $\mathsf{Inverse}_f$ queries into $\mathsf{Inverse}_{f'}$ queries.} For a given heavy hitter $y$, let $m_1(y), m_2(y), \ldots, m_{m(y)}(y)$ be the number $m_{k}(y) = \sum_{\ell = 1}^{k - 1} |f'^{-1}(y_\ell)|$ of elements that $f'$ routes to $y_1, \ldots, y_{k - 1}$ (so $m_1(y) = 0$). We will store the values of $m_k(y)$ for each heavy hitter $y$ and each $k \in [m(y)]$ in a hash table. To answer a $\mathsf{Inverse}_f(y, i)$ query, with $i \in [|f^{-1}(y)|]$, we will:
\begin{itemize}
\item Determine the largest number $j$ such that $m_j(y) < i$;
\item Return $\mathsf{Inverse}_{f'}(y_j, i - m_j(y))$.
\end{itemize}

The only slight challenge is the first step: how to determine $j$. Here, we can make use of the fact that, for each $j \in [m(y)]$, $|f'^{-1}(y_j)| = \Theta(\log^2 N)$. For each multiple $q \log^2 N \in [m(y)]$ of $\log^2 N$, store $\psi_y(q) = \argmax \{j \mid m_j(y) < q \log^2 N\}$. Then, for any number $i \in [m(y)]$, the largest number $j$ such that $m_j(y) < i$ is guaranteed to be in the range $[\psi_y(\lfloor i / \log^2 N \rfloor), \psi_y(\lfloor i / \log^2 N \rfloor) + O(1)]$. We can therefore determine $j$ in $O(1)$ time by simply checking which of the $O(1)$ candidates in the range is correct. 

Note that the $\psi_y(q)$ values take total space $O(\log N \cdot \sum_{\text{heavy hitter }y} |f^{-1}(y)| / \log^2 N) = O(N / \log N)$ bits, and a hash table storing them can be straightforwardly constructed in $O(N)$ time. Using this hash table, $\mathsf{Inverse}_f(y, i)$ queries can be answered in $O(T)$ time by evaluating the $\mathsf{Inverse}_{f'}(y_j, i - m_j(y))$ query instead. This completes the proof of Theorem \ref{thm:completeinversion-withhitters}.

\subsection{Dynamizing}\label{sec:dynamizing}

Finally, we are in a position to dynamize the data structure. In this section, we extend the data structure to support an $\mathsf{Update}(x, y)$ operation, which sets $f(x) = y$ for some $x, y \in [N]$. Rather than continuing to support the $\mathsf{Inverse}_f(y, i)$ queries from the previous section, the new data structure will support iterators: a function that allows one to iterate through all the elements of $f^{-1}(y)$ in $O(T)$ time per element (so long as no $\mathsf{Update}$ operations are performed during the iteration).
\begin{theorem}
    Let $f: [N] \to [N]$ be a function which we have $O(1)$-time oracle access to, and that is updated over time by $\mathsf{Update}(x, y)$ operations, and assume free randomness. For $T = O(\log N / \log \log N)$, we can construct in linear time a data structure that supports $O(T)$-time iterators, that uses space $O(N \log N / T)$ bits, and that supports $O(T)$-time $\mathsf{Update}(x, y)$ operations. Both the initial construction and each $\mathsf{Update}$ operation succeed with high probability in $N$.   
    \label{thm:dynamic}
\end{theorem}

The basic idea for the proof is to break the updates into phases of $O(N / T)$ operations. At the start of each phase we construct a new data structure (in linear time) using Theorem \ref{thm:completeinversion-withhitters}. Then, during the phase, we will keep track of three things, where $f_0$ denotes the function at the start of the phase, and $f$ denotes the current function:
\begin{enumerate}
\item For any $x \in [N]$ such that $f(x)$ has been updated, we store $f_0(x)$ in a hash table $\mathcal{A}$.
\item For each $y \in [N]$, keep track of all inverses $x \in f^{-1}(y)$ that have had $f(x)$ \emph{updated} to $y$ during the phase (call these \defn{fresh inverses} for $y$). This is stored in a hash table $\mathcal{B}$ that keeps, for each $y \in [N]$ that has any fresh inverses, a linked list of $y$'s fresh inverses. 
\item Finally, we keep an auxiliary data structure $\mathcal{U}$ that keeps track, for each $y \in [N]$, of which $i \in [|f_0^{-1}(y)|]$ have the property that $\mathsf{Inverse}_{f_0}(y, i)$ remains un-updated during the phase. The data structure $\mathcal{U}$ allows one to, for any given $y \in [N]$, iterate over such $i \in [|f_0^{-1}(y)|]$ in $O(1)$ time per $i$. 
\end{enumerate}

The first two data structures are straightforward to implement with $O(1)$-time operations and using $O(\log N)$ bits of space per entry. Because the phase lasts for $O(N / T)$ updates, the space for these data structures is $O(N \log N / T)$ bits. The third data structure is slightly more involved to implement. We will come back to it at the end of the section.

\paragraph{Supporting $O(T)$-time iterators. }Given the above data structures, we can support iterators (for the current function $f$) as follows. To iterate on the inverses of an element $y$, we first use $\mathcal{U}$ to iterate through the $i \in [|f_0^{-1}(y)|]$ such that $\mathsf{Inverse}_{f_0}(y, i)$ remains un-updated during the phase, and we return each such element $\mathsf{Inverse}_{f_0}(y, i)$ (in $O(T)$ time per element); then, we use $\mathcal{B}$ to iterate through the fresh inverses of $y$.  Note that, when calculating $\mathsf{Inverse}_{f_0}(y, i)$ for an un-updated $i$, we will need access to the \emph{original} function $f_0$ (even though the function we have oracle access to is $f$), which is why we also keep track of hash table $\mathcal{A}$. 

\paragraph{Supporting $O(T)$-time $\mathsf{Update}(x, y)$ operations. }As described so far, each update takes $O(1)$ time, and then each phase (which consists of $\Theta(N / T)$ updates) requires $\Theta(N)$ time to construct a new data structure. This amortizes to $O(T)$ time per update. 

We can further turn this amortized bound into a worst-case bound with the following essentially standard trick: Maintain two completely separate data structures $\mathcal{I}_1$ and $\mathcal{I}_2$, using the construction above, and swap between them as follows: Every $N / (2T)$ operations, we swap between which $\mathcal{I}_j$ is being used and which is currently being rebuilt. Each rebuild can then be spread across $\Theta(N/T)$ operations using $O(T)$ time per operation. While a given $\mathcal{I}_j$ is being rebuilt, updates are also being performed -- we keep track of these updates, and then as the final step of the rebuild we incorporate these updates into the $\mathcal{A}, \mathcal{B}, \mathcal{U}$ data structures for $\mathcal{I}_j$.\footnote{As we're doing this, additional updates to $f$ may occur, which are also incorporated live into the $\mathcal{A}, \mathcal{B}, \mathcal{U}$ data structures for $\mathcal{I}_j$.} Since each update to $\mathcal{A}, \mathcal{B}, \mathcal{U}$ takes $O(1)$ time, these $O(N / T)$ updates can collectively be fit into the $O(N)$ time allocated for the rebuild.

\paragraph{Implementing $\mathcal{U}$. } It remains to implement the data structure $\mathcal{U}$. To do this, we first abstract the problem that it must solve: 

\begin{definition}[The decremental prefix-set problem]
    We are given nonnegative integers \(a_1,\ldots,a_n\) satisfying
    \(\sum_i a_i=N\), and define sets \(S_i = [a_i]\) for each \(i \in [n]\).
    A \defn{decremental prefix-set} is a data structure that supports the following operations:
    \begin{itemize}
        \item A delete operation \(\mathsf{Delete}(i,j)\), which removes \(j\) from \(S_i\), doing
        nothing if \(j\notin S_i\);
        \item An iterator function $\mathsf{Iterate}(i)$ that creates an iterator
        $q$ supporting a $\mathsf{Next}(q)$ operation which iterates over the current
        contents of \(S_i\). The $\mathsf{Next}(q)$ operation is only legal to call if
        no delete operation has happened since the iterator was created.
    \end{itemize}
\end{definition}

In the case of $\mathcal{U}$, the sets $S_i$ correspond to the $i \in [|f_0^{-1}(y)|]$ such that $\mathsf{Inverse}_{f_0}(y, i)$ remains un-updated during the phase. The delete operation corresponds to an update to $f$, and the iterator corresponds to a way to iterate through the un-updated inverses of $y$.

The following proposition shows how to implement $\mathcal{U}$ using $O(N)$ bits of space and $O(1)$ time per operation. This then plugs into our construction to give Theorem \ref{thm:dynamic}. The proof makes use, in part, of a classic result of Raman, Raman, and Rao \cite{RamanRamanRao2007}, which gives a static $O(1)$-time rank-select data structure for a bitvector of length $N$, using $O(N)$ bits of space. 

\begin{proposition}
    \label{prop:succinct-decremental-sets}
    There is a data structure for the decremental prefix-set problem using
    \(O(N)\) bits of space and with time $O(1)$ per operation.
\end{proposition}
\begin{proof}
    We flatten all sets into a single length-\(N\) bitvector.  Let
    \[
    P_0=0
    \qquad\text{and}\qquad
    P_i=\sum_{k=1}^i a_k.
    \]
    The elements of \(S_i\) occupy the interval
    \[
    I_i=[P_{i-1},P_i)
    \]
    in the flattened bitvector, and the element \(j\in S_i\) corresponds to the
    position \(P_{i-1}+j-1\).
    
    We store the values \(P_i\) implicitly using the unary string
    \[
    U=0^{a_1}1\,0^{a_2}1\,\cdots\,0^{a_n}1.
    \]
    This string has length \(N + n\).  We equip \(U\) with the constant-time $O(N)$-bit
    rank/select structure \cite{RamanRamanRao2007}.  Taking positions in \(U\) to
    be one-indexed, we have
    \[
    P_i=\mathsf{select}_1(U,i)-i
    \]
    for \(i\ge 1\), and \(P_0=0\).  Thus the endpoints of any interval \(I_i\) can
    be recovered in \(O(1)\) time.
    
    Finally, we impose a block structure on the bitvector as follows.
    Let $b = (\log N) / 2$. Partition the flattened bitvector into
    $N / b$ blocks of $b$ bits each (with padding on the final block if necessary). 
    We store the current contents as an array
    \[
    A[0],A[1],\ldots,A[N/b - 1],
    \]
    where \(A[q]\) stores the current contents of the \(q\)-th block. 

    The point of these blocks is so that we can do the following: for each $i$ such that $S_i$ initially spans more than one block, we store a linked list of all blocks that
    contain at least one element of $S_i$ (this linked list is updated over time as deletions are performed).
    The heads of these linked lists are stored in a hash table $\mathcal{L}$, and the internal nodes are stored
    in the array $A$ (so $A[q]$ stores forward and backward pointers for any linked list it is in). 

    Having described what the data structure stores, we now turn to implementing the operations. To delete $j$ from $S_i$ we calculate $P_{i - 1}$ (using the rank/select data structure for $U$), we clear the appropriate bit in $A[P_{i-1} + j - 1]$, and if this 
    empties a block out, we remove that block from the appropriate linked list $\mathcal{L}[i]$. An iterator for $i \in [n]$ simply iterates through the blocks containing elements of $S_i$ (if there is more than one such block, we use the linked list $\mathcal{L}[i]$), and then uses
    standard bit-manipulation techniques to extract the elements of $S_i$ that the block encodes.\footnote{Alternatively, if one prefers not to use bit-manipulation, one can use lookup tables of size $O(2^b \log N) = o(N)$ to implement arbitrary operations on size-$b$ blocks.} All operations take $O(1)$ time.

    Finally, to bound space observe that the rank-select data structure for $U$ and the array $A$ each take $O(N)$ bits of space trivially. The hash table $\mathcal{L}$ also takes $O(N)$ bits of space since it stores at most $N/b = O(N / \log N)$ entries. Thus the total space usage is $O(N)$ bits.
\end{proof}

\subsection{Removing the Assumption of Free Randomness}\label{sec:random}

Finally, in this subsection, we show how to remove the assumption of free randomness. This leads us to the final result of the section:
\begin{restatable}{theorem}{thmmain}
    \label{thm:main}
    \thmmainbody
\end{restatable}

So far, we have assumed access to three random objects:
\begin{enumerate}
\item The permutation $\pi: [N] \rightarrow [N]$ that randomly assigns elements of $[N]$ to bins;
\item The hash function $r: [N] \rightarrow [N/B]$ that maps each element $y \in [N]$ to a random bin;
\item The hash function $\phi: [N] \rightarrow [0, 1]$ that maps each element $x \in [N]$ such that $f(x) = y$ is a heavy hitter to a random element $y_{\lceil m(y) \cdot \phi(x) \rceil}$ among $y_1, \ldots, y_{m(y)}$;
\end{enumerate}

We will show below how to simulate each of these objects using explicit hash-function and permutation constructions.

\paragraph{Simulating $\pi$. }Let us first review the property that $\pi$ must satisfy in order for our earlier analyses to hold. Consider the construction of some $D_k$, for some $k \in [\log^2 N]$, and consider a bin $b \in [N/B]$. Let $X$ be the number of target inverses that are placed in bin $b$ by the permutation $\pi$. In order for $\pi$ to be used in the analysis, we need \eqref{eq:Xch} to hold, meaning that, with high probability in $N$, we have
\begin{equation}
    X \ge \E[X] - O(\sqrt{B \log N}). 
    \label{eq:Xch2}
\end{equation}

To design a permutation $\pi$ that satisfies \eqref{eq:Xch2}, we permute the elements of $[N]$ as follows. We write the elements of $[N]$ in a table with $B$ rows and $N/B$ columns; we generate for each row $i \in [B]$ a random number $r_i \in [N/B]$; and we rotate the $i$-th row by $r_i$ positions to the right, mapping the $(j+1)$-th element of the $i$-th row to the $\bigl((j + r_i) \bmod (N/B)\bigr) + 1$-th element of the $i$-th row. Finally, we index the entries of the rows/columns so that each column is a bin. This construction requires $B \log N = \polylog(N)$ bits of space to store the random offsets $r_i$.

To prove \eqref{eq:Xch2}, define $A_1, A_2, \ldots, A_B$ so that $A_i$ is the indicator random variable for the event that bin $b$ contains a target inverse in row $i$. Each $A_i$ is fully determined by $r_i$, so the $A_i$s are mutually independent $0$-$1$ random variables whose sum is at most $B$. By a Chernoff bound, it follows that, with high probability in $N$, we have
$$\sum_{i=1}^B A_i \le \E\left[\sum_{i=1}^B A_i\right] + O(\sqrt{B \log N}).$$
Since $\sum_{i=1}^B A_i = X$, this gives \eqref{eq:Xch2}, as desired.

\paragraph{Helper machinery for simulating $r$ and $\phi$. }To simulate $r$ and $\phi$, we will need several pieces of machinery from the hash-function and concentration-bound literatures. 

From the hash-function literature, we will need the following result on limited-independence hash functions due to Siegel \cite{siegel2004universal}:
\begin{proposition}
\label{prop:hashfunction}
For any $N$ and $M$, there is a randomized construction of a hash function
$h:[N]\to[M]$ with the following properties. The description of $h$ consists of
\[
    N^{\delta}\log(M+N)
\]
random bits for a positive constant $\delta$ of our choice. With high probability in $N$ over the choice of this description,
the resulting distribution over hash functions is $\log^2 N$-wise independent.
Moreover, given the description, each value $h(x)$ can be computed in $O(1)$
time.
\end{proposition}

From the concentration-bound literature, we will need a Chernoff bound for random variables with limited independence. We will make use of the following (special case of a) result due to Schmidt, Siegel, and Srinivasan \cite{schmidt1995chernoff}:
\begin{proposition}[Chernoff bound with limited independence]
    Let \(n\ge 3\) and let $d$ be a positive constant. Then there positive constants $B_d$ and $C_d$ such that the following holds. 
    Let \(X_1,\ldots,X_n\in\{0,1\}\) be \(K\)-wise independent random variables, let $X = \sum_{i=1}^n X_i$, and let $\mu=\mathbb{E}[X]$.
    If
    $\mu \ge C_d^2 \log n$ and $K \ge B_d \log n$,
    then
    \[
        \Pr\!\left[
            |X-\mu| \ge C_d\sqrt{\mu\log n}
        \right]
        \le n^{-d}.
    \]
    \label{prop:chernofflimited}
    \end{proposition}

With these tools in place, we now proceed to simulating $r$ and $\phi$.

\paragraph{Simulating $r$. }Next, we review the property that $r$ must have in order for our earlier analyses to hold. Consider some time $t$ during the construction of some $D_k$, for some $k \in [\log^2 N]$, and consider a bin $b \in [N/B]$. Let $R$ be the total number of target inverses for data structure $D_k$. 
Let $Y$ be the number of target inverses in bin $b$ that have been placed into chains by time $t$. In order for $r$ to be used in the analysis, we need \eqref{eq:Ych} to hold, meaning that with high probability in $N$, we have
\begin{equation}
    Y \le R B / N - \Omega(B / \log^3 N) + O(\sqrt{B \log N}). 
    \label{eq:Ych2}
\end{equation}

Recall that, as described in Subsection \ref{sec:linear}, chains are constructed as follows: the first element of a chain is obtained by selecting a random bin $q_1$ and taking an (arbitrary) free target from that bin; then, given a member $q_i$ of the chain, we obtain the next member by selecting an (arbitrary) free target from bin $r(f(q_i))$. The number $Y$ is the total number of times that, during constructions of the chains in $D_k$, we select free targets from bin $b$. 

The main subtlety in constructing an explicit hash function $r$ that gives us \eqref{eq:Ych2} is that we must be careful about the following issue: as we construct a chain, we are applying $r$ multiple times, so that subsequent inputs to $r$ are determined, in part, by previous outputs of $r$. This prevents us from, say, directly implementing $r$ as a $\polylog N$-independent hash function. 

To rectify this issue, we will replace $r$ with a sequence of hash functions $r_1, r_2, \ldots, r_{T - 1}: [N] \to [N / (BT)]$, each of which is constructed by Proposition \ref{prop:hashfunction} to be $\log^2 N$-independent. We will follow the rule that:
\begin{enumerate}
\item The first element of each chain is from a random bin out of bins $[N/(BT)]$ (selected using true randomness during preprocessing);
\item The $i$-th element of each chain, for $i > 1$, is from bin $(i - 1) \cdot N/(BT) + r_{i - 1}(f(x))$, where $x$ is the $(i - 1)$-th member of the chain. 
\end{enumerate}
This introduces an additional data-structural issue, which is how to identify during query time which $r_i$ to use as we travel down a chain. This issue is easily resolved by adding an additional retrieval data structure (Proposition \ref{prop:retrieval}) mapping each element $y$ covered by any chain of $D_k$ to the $O(\log T)$-bit value $i$ indicating which $r_i$ should be used on $y$. This retrieval data structure uses $O(N_k (\log T + \log \log N)) = O(N \log \log N)$ bits where $N_k$ is the number of elements covered by chains in $D_k$. By construction, $\sum_k N_k \le N$, so the retrieval data structures collectively use $O(N \log \log N)$ bits.

With this modification in place, we can now prove \eqref{eq:Ych2}. Recall that, by design, the total number of elements placed in chains during the construction of $D_k$ is at most $R- \Omega(N / \log^3 N)$ (since $R$ is the total number of targets, and we stop building chains once there are $N / \log^3 N$ free targets remaining). Thus, the total number of chains is at most $R / T - \Omega(N / (T \log^3 N))$. Suppose bin $b$ is of the form $(i - 1) \cdot N/(BT) + \ell$ for some $\ell \in [N / (BT)]$. This means that chains only use bin $b$ for their $i$-th element.

If $i = 1$, then chains only use bin $b$ for their first element. These first elements are selected from iid uniformly random bins in the range $[N/T]$. Since the number of chains is at most $R/T - \Omega(N / (T \log^3 N))$, we may conclude that $Y$ is dominated by a sum of independent Bernoulli random variables with total mean 
$$\bigl(R/T - \Omega(N / (T \log^3 N))\bigr) \cdot (T B / N) = RB/N - \Omega(B / \log^3 N).$$
The bound \eqref{eq:Ych2} then follows by a Chernoff bound.

If $i > 1$, then the chains use bin $b$ only as their $i$-th element. These $i$-th elements are selected by computing $r_{i - 1}(f(x))$ for each of the elements $x$ that appear as the $(i - 1)$-th element of a chain. By construction, the set $V$ of such elements $x$ is fully independent of $r_{i - 1}$. The size of $V$ is at most $R/T - \Omega(N / (T \log^3 N))$ (the maximum number of chains), and each element $x \in V$ has probability at most $TB/N$ of satisfying $r_{i - 1}(f(x)) = b$. Since $r_{i - 1}$ is a $\log^2 N$-wise independent hash function, we may conclude that $Y$ is a sum of at most $R/T - \Omega(N / (T \log^3 N))$ $\log^2 N$-wise independent $0$-$1$ random variables with total mean at most 
$$\bigl(R/T - \Omega(N / (T \log^3 N))\bigr) \cdot (T B / N) = RB/N - \Omega(B / \log^3 N).$$
Applying Proposition \ref{prop:chernofflimited} to $Y$, with $\mu = RB/N - \Omega(B / \log^3 N)$, $n = N$, and $K = \log^2 N$ gives \eqref{eq:Ych2}.

\paragraph{Simulating $\phi$. } Finally, we may simulate $\phi$ by directly implementing $\phi$ as a $\log^2 N$-independent hash function, implemented using Proposition \ref{prop:hashfunction}. (To simulate an output range of $[0, 1]$ for $\phi$, we can use a large $\poly(N)$ size range, and renormalize to $[0,1]$.) The only role of $\phi$ is to ensure that, for each heavy hitter $y$, if we define $m (y)$ as in \eqref{eq:mdefn}, then for each $i \in [m(y)]$, we have
\begin{equation}\Omega(\log^2 N) \le |f'^{-1}(y_i)| \le \log^2 N.
    \label{eq:phi}
\end{equation}
By replacing the Chernoff bounds in the proof of Lemma \ref{lem:phi} with Proposition \ref{prop:chernofflimited}, we can recover \eqref{eq:phi} in exactly the same way as in Lemma \ref{lem:phi}. 

Putting the pieces together, we have proven Theorem \ref{thm:main}.
\section{Application to Unordered Graphs}\label{sec:unordered-graphs}
\label{sec:dynamic-unordered-graphs}

We now describe an application of Theorem \ref{thm:main} to dynamic graph
representations. Farzan and Munro~\cite{FarzanMunro08}
considered the corresponding static problem: they showed how to represent an
arbitrary labeled graph using space within a $(1+\epsilon)$ multiplicative
factor of the information-theoretic optimum, while supporting adjacency
and neighborhood queries in $O(\epsilon^{-1})$ time.\footnote{The 
original paper does not state the precise dependency on $\epsilon^{-1}$, but 
a careful reading of the proof reveals that it is $O(\epsilon^{-1})$.}

In this section, we give a dynamic version of the same result (in the 
parameter regime $m = n^{2 - \Omega(1)}$ and $\epsilon = \Omega(\log \log n / \log n)$), allowing 
the graph to be updated by edge insertions and deletions. Along the way,
we also give a new construction for the static setting, which reduces the
problem to function inversion.

\paragraph{Problem setup. } Throughout this section, a graph means a simple
 undirected graph on vertex set
$[n]$. We write $G=([n],E)$ and $m=|E|$, where $m$ always denotes the current
number of edges. Let $M={n\choose 2}$ be the number of possible edges. The
information-theoretic optimum for representing an $m$-edge graph is
$\log {M\choose m}$ bits.


A \defn{dynamic graph data structure} must support the following operations:
\begin{itemize}
    \item $\mathsf{Adjacent}(i,j)$, which returns whether $\{i,j\}\in E$;
    \item $\mathsf{Insert}(i,j)$, which inserts $\{i,j\}$ into $E$, doing
    nothing if the edge is already present;
    \item $\mathsf{Delete}(i,j)$, which deletes $\{i,j\}$ from $E$, doing
    nothing if the edge is not present;
    \item $\mathsf{Iterate}(i)$, which creates an iterator $q$ for the current
    neighbors of $i$;
    \item $\mathsf{Next}(q)$, which returns the next neighbor in the iterator,
    or returns $\bot$ if all neighbors have already been returned.
\end{itemize}
The iterator returned by $\mathsf{Iterate}(i)$ returns every current neighbor of
$i$ exactly once, in an arbitrary order. It is only legal to call
$\mathsf{Next}(q)$ if no insertion or deletion has occurred since $q$ was
created.

Finally, we say that such a data structure is \defn{$(1+\epsilon)$-compact} if, at all times,
it uses at most
\begin{equation}
    (1+\epsilon)\log {M\choose m}
    + O((n + m) \log\log n) + O((n + m) \epsilon \log n)
    \label{eq:compact}
\end{equation}
bits. Note that, so long as $m \le n^{2 - \Omega(1)}$, this space bound can be rewritten as $(1 + \Theta(\epsilon)) \log \binom{M}{m}$. 

With these definitions in place, the main result of this section is the following:

\begin{restatable}{theorem}{thmunorderedgraphs}
    \label{thm:dynamic-unordered-graphs}
    \thmunorderedgraphsbody
\end{restatable}

To prove Theorem \ref{thm:dynamic-unordered-graphs}, we will proceed as follows. First,
we present some existing data-structural machinery that will be useful in our construction.
Then, we describe a static solution to the problem, which shows how to reduce the problem 
to function inversion. Finally, using dynamic function inversion instead, we are able
to obtain Theorem \ref{thm:dynamic-unordered-graphs}.

\subsection{Helper Machinery}



Our constructions will make use of three pieces of data-structural machinery. 
The first is a data structure that allows one to space-efficiently store an 
array of variable-length objects (Theorem 6 of \cite{CramCompressedRAM}). 

\begin{proposition}[Array for variable-length objects~\cite{CramCompressedRAM}]
    Assume machine words of size $w = \Omega(\log n)$ bits, and let $B[1],\ldots,B[r]$ be bit strings of
    length at most $O(w)$. The strings can be stored in
    $
        \sum_i |B[i]|+O(w^4+r\log \log n)
    $
    bits, while supporting $O(1)$-time accesses and updates to each $B[i]$ (where
    an update may change $|B[i]|$). 
    \label{prop:word-sized-variable-blocks}
\end{proposition}

The second piece of machinery is what is known as a monotone minimal perfect hash function (Theorem 3.1 of \cite{monotoneminimal}). 
\begin{proposition}[Monotone minimal perfect hashing]
    Let $S$ be a static set of $r$ keys from a universe of size $n^{O(1)}$, and
    suppose the keys of $S$ are ordered lexicographically. There is a data
    structure using $O(r\log\log n)$ bits that returns, for any $x\in S$, the
    rank of $x$ in $S$ in $O(1)$ time. For $x\notin S$, the data structure has
    no guarantee on what it returns. The data structure can be built in linear time
    with high probability in $r$.
    \label{prop:mmphf}
\end{proposition}

Note that monotone-minimal perfect hash functions are \emph{static} data structures, 
and we will only use Proposition \ref{prop:mmphf} on static sets. 

Finally, we will also make use of a classic result on static rank-select data structures \cite{RamanRamanRao2007}. Note that we also used this result earlier in the proof of Proposition \ref{prop:succinct-decremental-sets}.
\begin{proposition}[Rank-Select on a Bit-Array]
Given a set $S \subseteq [m]$, one can construct in $O(m)$ time a static rank-select data structure supporting rank and select in time $O(1)$ and using $O(m)$ bits of space.
\label{prop:rankselect}
\end{proposition}

With these tools in place, we can now continue to the construction.

\subsection{Encoding the Graph $G$ as a Function} \label{sec:encoding-f_G}

Our high-level approach to proving Theorem \ref{thm:dynamic-unordered-graphs}
will be to first encode the graph $G$ as a function, and to then apply
dynamic function inversion to that function. 

\paragraph{Interpreting $G$ as a function. } We begin by defining the function $f_G$ (in the static setting).
Define the \defn{owner} of an edge $\{i, j\}$ as the smaller endpoint
$\min(i, j)$. For each node $i$, let $J_i$ be the set of edges with owner $i$, 
let $d_i = |J_i|$ and let $j_{i,1} < j_{i,2} < \cdots < j_{i,d_i}$ be the other 
endpoints (in increasing order) of the edges in $J_i$.
Define a function $\overline{f}_G$ which maps each pair $(i, k)$, $k \in [d_i]$, 
to the endpoint $j_{i, k}$. 

If we could store both $\overline{f}_G$ and its inverse $\overline{f}_G^{-1}$, then we 
could recover all the neighbors of a node $i$ by iterating through both $\{\overline{f}_G((i, k)) : k \in [d_i]\}$ and $\overline{f}_G^{-1}(i)$.

There are two high-level challenges in doing this. First, we need to efficiently store the function $\overline{f}_G$ 
by itself (and in a way that also supports adjacency queries), while using space very close to $\log {M\choose m}$ bits. Second, we must translate $\overline{f}_G$ into a function $f_G$ that maps $[m] \rightarrow [n]$
(instead of mapping $\{(i, k) : i \in [n], k \in [d_i]\}$ to $[n]$). If we can do this, then the rest of the construction becomes relatively straightforward.

In this section, we show how to address both of the challenges above, storing $\overline{f}_G$ efficiently, and translating it into an $f_G: [m] \rightarrow [n]$.

\paragraph{Encoding $\overline{f}_G$ efficiently. }
We now describe how to store $\overline{f}_G$ efficiently.

We will refer to pairs $(i,k)$, with $1\le k\le d_i$, which are the valid
inputs to $\overline{f}_G$, as \defn{$\overline{f}$-inputs}, and we will 
think of these inputs as being sorted lexicographically, first by $i$ and then by $k$. 
Each $\overline{f}$-input $(i,k)$ corresponds to the edge
 $(i,j_{i,k})$. Define the \defn{global rank} of a $\overline{f}$-input $(i,k)$ (or, equivalently, the edge $(i,j_{i,k})$) to be
\[
    \operatorname{rank}(i,k)=k+\sum_{h<i} d_h.
\]
This means that, when sorted in increasing order, the $\overline{f}$-inputs occupy ranks $1,\ldots,m$.

We store the degree sequence $d_1,\ldots,d_n$ using the unary string
$0^{d_1}1\,0^{d_2}1\,\cdots\,0^{d_n}1$, equipped with a constant-time
$O(m)$-bit rank-select data structure \cite{RamanRamanRao2007} (Proposition \ref{prop:rankselect}).
 This lets us convert in $O(1)$ time between a
global rank $p\in[m]$ and the $\overline{f}$-input $(i,k)$ of rank $p$.

Next we describe how to map each $\overline{f}$-input $(i,k)$ to $j_{i, k}$.  
For each list $J_i$, we divide the elements $j_{i,1},\ldots,j_{i,d_i}$ into 
blocks of size $\log n$ (where the last block may be smaller). We then construct
an array of variable-length objects $A[1],\ldots,A[m]$ using Proposition \ref{prop:word-sized-variable-blocks}, 
where the entries are determined as follows:
\begin{itemize}
    \item For each block $B$, the first entry $j_{i, k}$ in that block is stored
    directly in $A[\operatorname{rank}(i, k)]$.
    \item For each block $B$ with first-entry $j_{i, k^*}$, and for each other
    entry $j_{i, k}$ in the block, we set $A[\operatorname{rank}(i, k)] = j_{i, k} - j_{i, k^*}$.
\end{itemize}

We can therefore recover $\overline{f}_G((i, k))$ in $O(1)$ time by 
first computing $\operatorname{rank}(i, k)$ and $\operatorname{rank}(i, k^*)$,
then returning $A[\operatorname{rank}(i, k^*)]$ if $k = k^*$, and 
$A[\operatorname{rank}(i, k)] + A[\operatorname{rank}(i, k^*)]$ otherwise.

Moreover, since we can translate between global ranks and $\overline{f}$-inputs in $O(1)$ time, 
we can define $f_G: [m] \rightarrow [n]$ to map global ranks $p \in [m]$ to $\overline{f}_G((i, k))$ 
where $(i, k)$ is the $\overline{f}$-input of rank $p$. Thus, we have not only shown how to 
store $\overline{f}_G$, but we have also translated it into a function $f_G: [m] \rightarrow [n]$.

Finally, so that our data structure supports adjacency queries, we also store a monotone minimal perfect hash function (constructed using Proposition \ref{prop:mmphf}) 
that maps each edge $(i, j)$, $i < j$, to its global rank ($\operatorname{rank}(i, k)$, where $j = j_{i, k}$).
This allows us to answer adjacency queries in $O(1)$ time: to check if $\{i, j\} \in E$, where $i < j$, we simply query the monotone minimal perfect hash function on $(i, j)$, to get a rank $p$, and check if $f_G(p) = j$.

\paragraph{Bounding the space used to store $\overline{f}_G$ and $f_G$. }
We now argue that the data structures above use total space
\begin{equation}
    \log\binom{M}{m} + O((n + m)\log\log n).
    \label{eq:space-used-to-store-f_G}
\end{equation}
bits. Since the monotone minimal perfect hash function uses $O(m\log\log n)$ bits,
the rank-select data structure uses $O(m)$ bits,
and the array of variable-length objects uses $\sum_t |A[t]| + O((n + m) \log \log n)$ bits,
it suffices to show that
\begin{equation}
    \sum_t |A[t]| \le \log\binom{M}{m} + O((n + m) \log \log n).
    \label{eq:space-used-to-store-A}
\end{equation}

First consider the $A[t]$s used to encode the blocks of a given $J_i$. 
Let $j_{i, k^*_1}, j_{i, k^*_2}, \ldots, j_{i, k^*_s} \in [n]$ be the
first elements of each of the blocks in $J_i$, and define $a_1, \ldots, a_s$ 
so that $a_j$ is the sum of the $|A[t]|$s corresponding to the block starting at $j_{i, k^*_j}$.
Then, for every $j \in [1, s)$, we have that
$$a_j \le \log n \cdot \log (j_{i, k^*_{j + 1}} - j_{i, k^*_j}),$$
and for $j = s$ we have that 
$$a_s \le z \log n,$$
where $z$ is the number of elements in the last block. Summing over the blocks in $J_i$,
\begin{align*}
    \sum_j a_j & \le \sum_{j=1}^{s-1} \log n \cdot \log (j_{i, k^*_{j + 1}} - j_{i, k^*_j}) + z \log n \\
    & \le (s - 1) (\log n) \log \frac{n - i}{s - 1} + z \log n \tag{since $j_{i, k^*_{s}} - j_{i, k^*_{1}} \le n - i$ and by Jensen's inequality}\\
    & = d_i \log \frac{n - i}{d_i} + z \log d_i \tag{since $(s - 1) (\log n) + z = d_i$} \\
    & = d_i \log n - d_i \log d_i + O(d_i \log \log n) \\
    & = \log \binom{n - i}{d_i} + O(d_i \log \log n) \tag{by Stirling's approximation}
\end{align*}
Summing over the $J_i$s, it follows that
\begin{align*}
    \sum_t |A[t]| & \le \sum_i \log \binom{n - i}{d_i} + O(\sum_i d_i \log \log n) \\
    & = \sum_i \log \binom{n - i}{d_i} + O((n + m) \log \log n) \tag{since $\sum_i d_i = m$} \\
    & \le \log\binom{M}{m} + O((n + m) \log \log n) \tag{ by the identity $\sum_{i = 1}^n \log\binom{x_i}{y_i} \le \log\binom{\sum_{i = 1}^n x_i}{\sum_{i = 1}^n y_i}$},
\end{align*}
This establishes \eqref{eq:space-used-to-store-A}, and therefore \eqref{eq:space-used-to-store-f_G}.

\subsection{A Static Warmup}

We now show how to use the data structure above to build a $(1 + \epsilon)$-compact
(but static) representation of $G$. 

The data structure as described already allows us to:
\begin{enumerate}
\item Answer adjacency queries in $O(1)$ time.
\item Translate between global ranks and $\overline{f}$-inputs in $O(1)$ time.
\item Evaluate $f_G$ in $O(1)$ time.
\item Iterate, for a given vertex $i$, through the neighbors $j > i$ in $O(1)$ time per neighbor.
\end{enumerate}

Using Theorem \ref{thm:main}, applied to $f_G$, we can store an additional $O(\epsilon (n + m) \log n)$-bit
data structure that allows us to iterate through the inverses $f_G^{-1}(i)$ in $O(1)$ time per inverse,
for each $i \in [n]$. Converting each such inverse $p$ to a $\overline{f}$-input $(j, k)$ lets us
recover the neighbors $j < i$ of $i$. Thus, we can iterate through all the neighbors of $i$ in $O(1)$ time per neighbor.

The total space is 
$$\log\binom{M}{m} + O((n + m)\log\log n) + O(\epsilon (n + m) \log n)$$
bits and the data structure supports $O(\epsilon^{-1})$-time adjacency queries and unordered neighbor iterators. 
It is also straightforward to construct the data structure in $O(n + m)$ time. Thus, we get the static version of Theorem \ref{thm:dynamic-unordered-graphs}.

\subsection{Dynamizing}\label{sec:dynamic}

We now prove Theorem \ref{thm:dynamic-unordered-graphs}. The basic idea is 
to use dynamic function inversion rather than static function inversion.
We will also need to dynamize our representation of the graph $G$ itself.
We will do this by following the same phase-based structure as Theorem \ref{thm:main}: at the
beginning of a phase we build a static representation of $G$,
and during the phase we maintain a lower-order amount of metadata describing
the updates that have occurred since the rebuild.

Let $m_0$ be the number of edges at the beginning of a phase. The phase lasts
for $\Theta(\epsilon (n+m_0))$ updates. For a given phase, we use $G_0$ to denote the graph at the beginning of the phase, and $f_0$ to denote the function $f_{G_0}$.
At the beginning of each phase, we build the static representation of $f_0$ (as described in Section \ref{sec:encoding-f_G}). We will argue at the end of the section that, between phases, these rebuilds can 
be performed in a way that maintains the $(1+\epsilon)$-compactness guarantee, and where the amortized
time per update remains $O(\epsilon^{-1})$. 

During the phase, we distinguish between \defn{old edges}, which are edges of
$G_0$, and \defn{fresh edges}, which are edges inserted during the current phase
and still present in the graph. An old edge becomes \defn{touched} the first
time it is deleted during the phase. Once an old edge is touched, its old copy
is retired permanently for the rest of the phase. If the same edge is later
reinserted during the phase, it is represented as a fresh edge. This convention
ensures that no edge is ever returned once from the old static representation
and once from the fresh-edge metadata.

During a given phase, in addition to storing the static representation of $f_0$, we store as metadata: 
\begin{enumerate}
    \item a dictionary of touched old edges;
    \item linked lists of fresh edges grouped by owner; \label{item:fresh-edge-lists}
    \item linked lists of fresh edges grouped by non-owner endpoint; \label{item:fresh-edge-lists-by-non-owner}
    \item a dictionary of fresh edges (with pointers into the linked lists in items \ref{item:fresh-edge-lists} and \ref{item:fresh-edge-lists-by-non-owner}, so that fresh edges can be deleted from those lists in constant time);
    \item a decremental prefix-set structure (as given by Proposition \ref{prop:succinct-decremental-sets}) keeping track of which $\overline{f}_{G_0}$-inputs $(i, k)$ correspond to old edges that have not yet been touched; \label{item:decremental-prefix-set}
    \item an inverse-data structure (given by Theorem \ref{thm:main}) for the function $f_0'$ defined for any old edge of rank $r$ to be 
    $f'_0(r) = f_0(r)$ if $r$ has not been touched, and $f'_0(r) = n+1$ otherwise. Note that the oracle for $f_0'$ can be implemented in $O(1)$ time using the static representation of $f_0$ and the decremental prefix-set structure in item \ref{item:decremental-prefix-set}.
\end{enumerate}
These metadata are also rebuilt between phases.

\paragraph{Bounding space usage during a phase.}
We now bound the space used by the data structure during a given phase (not counting space used to perform 
the rebuilds between phases, which we will come back to at the end of the section).

The phase contains only $O(\epsilon (n+m_0))$ updates. Thus the touched-edge and
fresh-edge dictionaries, together with the fresh-edge linked lists, store only
$O(\epsilon (n+m_0))$ edges and use $O(\epsilon (n+m_0)\log n)$ bits. 
The decremental prefix-set structure uses $O(m_0)$ bits, and the
dynamic function-inversion data structure uses $O(\epsilon (n+m_0)\log n)$ bits. Altogether,
the phase metadata fits in
\[
    O((n+m_0)\epsilon\log n + (n+m_0)\log\log n)
\]
bits. Since the representation of $f_0$ itself uses 
$$\log\binom{M}{m_0} + O((n + m_0)\log\log n) + O(\epsilon (n + m_0) \log n)$$
bits, the total space used at any point during the phase is 
$$\log\binom{M}{m_0} + O((n + m_0)\log\log n) + O(\epsilon (n + m_0) \log n) + O((n + m_0)\epsilon\log n + (n + m_0)\log\log n).$$
Since $m_0 = m \pm O(\epsilon (n + m_0))$, this is 
$$\log\binom{M}{m} + O((n + m)\log\log n) + O(\epsilon (n + m) \log n),$$
as desired.

\paragraph{Adjacency queries.}
Next, we describe how to support each type of operation in $O(\epsilon^{-1})$ time. 
To answer $\mathsf{Adjacent}(i,j)$, assume without loss of generality that
$i<j$. We first check the fresh-edge dictionary. If $\{i,j\}$ is present there,
then the answer is yes. Otherwise, we check if $(i, j)$ was present in $G_0$
(using the static representation of $f_0$), and further check if $(i, j)$
has been touched during the current phase. If it was present in $G_0$ and has 
not been touched, then it is present, and otherwise it is not. 
The adjacency query takes $O(1)$ time.

\paragraph{Insertions.}
To perform $\mathsf{Insert}(i,j)$, we first check whether the edge is already
present, using the adjacency procedure above. If the edge is already present,
we do nothing. Otherwise, assume without loss of generality that $i<j$, so that
$i$ is the owner. We insert the edge into the fresh-edge dictionary, into the
fresh owned-edge list for $i$, and into the fresh non-owner list for $j$. This
takes $O(\epsilon^{-1})$ time.

\paragraph{Deletions.}
To perform $\mathsf{Delete}(i,j)$, assume without loss of generality that
$i<j$. If the edge is fresh, we remove it from the fresh-edge dictionary and
from its two fresh-edge linked lists. Otherwise, if the edge is present in the
old static representation and has not yet been touched, we mark it as touched.

In the latter case, we also update the decremental prefix-set structure, 
and the function $f'_0$ to reflect the fact that the edge has been touched. 
In total, the deletion takes $O(\epsilon^{-1})$ time.

\paragraph{Neighbor iterators.}
To iterate through the neighbors of a vertex $i$, we combine four sources:
\begin{enumerate}
    \item old untouched edges owned by $i$;
    \item old untouched edges owned by a smaller vertex and pointing to $i$;
    \item fresh edges owned by $i$;
    \item fresh edges whose non-owner endpoint is $i$.
\end{enumerate}
The first source can be iterated through using the data structure for $f_0$
along with the decremental prefix-set structure. The second source can be 
iterated through by calculating the inverses ${f'_0}^{-1}(i)$. The last two sources 
can be iterated through using the fresh-edge linked lists.

The combined iterator returns every
current neighbor exactly once, with $O(\epsilon^{-1})$ time per iterator operation.

\paragraph{Rebuilding.}
At the end of a phase, the static representation of $f_0 = f_{G_0}$
must be rebuilt to instead represent $f_1 = f_{G_1}$, where $G_1$ is the graph 
at the end of the phase. 

We now discuss how to perform this rebuild without temporarily blowing up 
the space usage of the data structure. Specifically, we wish to perform the 
rebuild in time $O(n+m_0)$, while preserving $(1 + \epsilon)$-compactness
as the rebuild occurs.

The rebuild can be accomplished by:
\begin{enumerate}
\item Sorting the changes that must be applied to order of edge rank;
\item Scanning through the array $A[1], \ldots, A[m_0]$, from left to right,
clearing out the entries as we go (so some prefix is all null entries), and 
filling in the entries of our new array $A'[1], \ldots, A'[m_1]$. 
\item Rebuilding the rank-select data structure, the monotone minimal perfect hash function,
the metadata structures, each in the straightforward $O(n + m_0)$-time manner. Note that
these all use little enough space that naive rebuilds are not a problem space-wise. 
\end{enumerate}

Critically, as we clear out a prefix of the array $A$, and fill in a prefix of the array $A'$,
we can argue that the \emph{total} space between them is never too large. Indeed, the amount of
space used by a prefix of $A$ (before being cleared out) and the same prefix of $A'$ (after being filled in)
is always the same up to $\pm O(\epsilon (n + m_0) \log n)$ bits. So, the total blowup in space during
the rebuild is at most an additive $O(\epsilon (n + m_0) \log n)$ bits, which does not violate the $(1+\epsilon)$-compactness guarantee.

The rebuild takes total time $O(n+m_0)$. Since the previous phase consisted of $O(\epsilon (n+m_0))$ updates,
this amortizes to $O(\epsilon^{-1})$ time per update. This completes the proof of Theorem \ref{thm:dynamic-unordered-graphs}.

\section{Dynamic Function Inversion in All Parameter Regimes}\label{sec:fiat-naor-dynamic}

Finally, in this section, we show that our dynamization techniques from Section \ref{sec:main} can also be applied to the classic Fiat-Naor construction. The main theorem of the section (using the same terminology as defined in Section \ref{sec:dynamic}) is the following:

\begin{restatable}{theorem}{thmfiatnaor}
        \label{thm:dynamic-fiat-naor}
        \thmfiatnaorbody
\end{restatable}
Note that, in Theorem \ref{thm:dynamic-fiat-naor}, and throughout the section, $\tilde{O}(\cdot)$ hides $\polylog N$ terms. These terms are considered negligible for the classic Fiat-Naor construction, where the main focus is on the case of $t = N^{\Omega(1)}$. 

The classic construction by Fiat and Naor \cite{FiatNa91,FiatNa00} has a notion of a \defn{heavy hitter}, which is any element $y \in [N]$ such that $|f^{-1}(y)| \ge t$. A useful feature of the classic construction, which is left implicit in the classical analysis \cite{FiatNa91,FiatNa00}, is that, if all elements are \emph{non-heavy hitters}, the function-inversion query can be adapted to return \emph{all of} $f^{-1}(y)$ in $\tilde{O}(t^3)$ time, not just a single inverse. This is captured in the following variation of the Fiat and Naor result, due to Bibbens, Borevitz, and McCauley \cite{BibbensBorevitzMcCauley2026TextIndexingMismatches}:

\begin{theorem}
        Let $f: [N] \to [N]$ be a function which we have $O(1)$-time oracle access to, and with the feature that $|f^{-1}(y)| \le t$ for all $y \in [N]$. Then, we can construct in $\tilde{O}(N)$ time a data structure that returns \emph{all} inverses of a given element $y \in [N]$ in $\tilde{O}(t^3)$-time, that uses space $O(N \log N / t)$ bits, where the construction succeeds with high probability in $N$. 
        \label{thm:static-fiat-naor-all-inverses}
\end{theorem}

In the rest of the section, we show how to use Theorem \ref{thm:static-fiat-naor-all-inverses} in a black-box fashion to prove Theorem \ref{thm:dynamic-fiat-naor}. Throughout, we suppose without loss of generality that $t \ge \log^2 N$ (since we are ignoring polylogarithmic terms).

\paragraph{Step 1: Recovering all inverses. } The first step is to modify the static construction so that it can also recover all inverses for heavy hitters (in time at most $\tilde{O}(t^3)$ per inverse). We do this with the same approach as in Section \ref{sec:heavy-hitters-to-non-heavy-hitters}.

For each heavy hitter $y$, let 
\begin{equation}
    m(y) = \left\lceil \frac{2|f^{-1}(y)|}{t} \right\rceil,
    \label{eq:mdefn-fn}
\end{equation}
and define $y_1, \ldots, y_{m(y)}$ to be the integers in the interval
$$\left( N + \sum_{\text{heavy hitter }y' < y} m(y'),\; N + \sum_{\text{heavy hitter }y' \le y} m(y') \right].$$
Let $\phi:[N] \rightarrow [0, 1]$ be a $\log^2 N$-wise independent hash function constructed using Proposition \ref{prop:hashfunction} to support $O(1)$ evaluation time and take space at most, say, $N^{1/3}$ bits (and simulating output range $[0, 1]$ by mapping to some large polynomial-size range and renormalizing); we will condition on the high-probability event from Proposition \ref{prop:hashfunction} that $\phi$ is $\log^2 N$-wise independent.
Let $N' = \sum_{\text{heavy hitter }y} m(y)$, and define a new function $f':[N'] \rightarrow [N']$ by 
$$
f'(x) = \begin{cases}
0 & \text{if } x > N, \\
f(x) & \text{if } x \in [N] \text{ and } f(x) \text{ is not a heavy hitter}, \\
f(x)_{\lceil m(f(x)) \cdot \phi(x) \rceil} & \text{if } x \in [N] \text{ and } f(x) \text{ is a heavy hitter}.
\end{cases}
$$
As in Section \ref{sec:heavy-hitters-to-non-heavy-hitters}, the basic idea behind $f'$ is that, for each heavy hitter $y$, it reroutes each element $x \in f^{-1}(y)$ to a random element among $y_1, \ldots, y_{m(y)}$. By distributing the preimage of $y$ among $m(y)$ new elements, we can eliminate heavy hitters.

\begin{lemma}
With high probability in $N$, $f'$ has no heavy hitters. Moreover, for each heavy hitter $y$ and each $i \in [m(y)]$, we have $|f'^{-1}(y_i)| = \Theta(t)$.
\label{lem:phi2}
\end{lemma}
\begin{proof}
It suffices to show that, for each heavy hitter $y$ (with respect to $f$), and for each $i \in [m(y)]$, we have $|f'^{-1}(y_i)| \in [\Omega(t), t]$. Notice that 
$$|f'^{-1}(y_i)| = \sum_{x \in f^{-1}(y)} \mathbbm{1}\!\left[\left\lceil m(y) \cdot \phi(x) \right\rceil = i\right]$$
is a sum of $|f^{-1}(y)|$ $\log^2 N$-wise independent 0-1 random variables with total mean $|f^{-1}(y)| / m(y) \in [0.25 t, 0.5 t]$. By a Chernoff bound for $\log^2 N$-wise independent random variables (Proposition \ref{prop:chernofflimited}), and using the fact that $t \ge \log^2 N$, it follows that $\Pr[\,|f'^{-1}(y_i)| \not\in [0.1 t, 0.9 t]\,] \le e^{-\Omega(\log^2 N)}$. Union bounding over all heavy hitters $y$ and all $i \in [m(y)]$, the lemma follows. 
\end{proof}

As in Section \ref{sec:heavy-hitters-to-non-heavy-hitters}, we have that $N'$ is only slightly larger than $N$. By design, 
$$N'- N \le \sum_{\text{heavy hitter }y} m(y) \le \sum_{\text{heavy hitter }y} |f^{-1}(y)| / t = O(N / t).$$
Thus we can afford to store $O((N' - N) \polylog N)$ bits of metadata in our data structure. 

With this in mind, it is straightforward to construct $f'$ in $O(N)$ time, and using $\tilde{O}(N/t)$ bits, as follows. We first randomly sample a $(\log^2 N / t)$-fraction of elements and declare for each sampled element $x$ that $f(x)$ is a potential heavy hitter. Note that this sampling step will, with high probability in $N$, find all heavy hitters, since the probability of a heavy-hitter being missed is at most
$$(1 - t / N)^{(N/t)\log^2 N} \le e^{-\Omega(\log^2 N)}.$$
Next, we perform a pass on the elements of $[N]$ to calculate $|f^{-1}(y)|$ for each candidate heavy-hitter $y$ (this uses $\tilde{O}(N /t)$ bits). This lets us determine the true heavy hitters, and their $m(y)$ values (which we store in a hash table). Finally, for each heavy hitter $y$ and each $i \in [m(y)]$, we store a hash table mapping $(y, i)$ to $y_i$ (and vice-versa). This results in an $\tilde{O}(N / t)$-bit data structure which can be used to translate calls to $f'$ into calls to $f$ and vice-versa. 

Finally, we can use $f'$ to iterate through the inverses of a heavy hitter $y$ in $\tilde{O}(t^3)$ time per inverse by simply iterating through each $y_i$, and then using Theorem \ref{thm:static-fiat-naor-all-inverses} to get the entire set $f^{-1}(y_i)$ in $\tilde{O}(t^3)$ time.

\paragraph{Step 2: Dynamizing the construction. }Having modified the static construction to keep track of all inverses for every element, we can now modify it to support dynamic updates. We do this by using the same high-level approach as in Section \ref{sec:dynamizing}. 

As in Section \ref{sec:dynamizing}, we will proceed in phases. At the start of each phase we construct a new data structure (in $\tilde{O}(N)$ time) using the construction above. Then, during the phase, we will keep track of the following three things, where $f_0$ denotes the function at the start of the phase, and $f$ denotes the current function. Note that only the third bullet point differs from Section \ref{sec:dynamizing}.
\begin{enumerate}
\item For any $x \in [N]$ such that $f(x)$ has been updated, we store $f_0(x)$ in a hash table $\mathcal{A}$.
\item For each $y \in [N]$, keep track of all inverses $x \in f^{-1}(y)$ that have had $f(x)$ \emph{updated} to $y$ during the phase (call these \defn{fresh inverses} for $y$). This is stored in a hash table $\mathcal{B}$ that keeps, for each $y \in [N]$ that has any fresh inverses, a linked list of $y$'s fresh inverses. 
\item Finally, we keep an auxiliary data structure $\mathcal{U}$ that keeps a linked list, for each $y \in [N]$, of which $i \in [m(y)]$ have the property that not all $f'$-inverses of $y_i$ have been updated.
\end{enumerate}

The first two data structures are straightforward to implement with $O(1)$-time operations and using $O(\log N)$ bits of space per entry. Because the phase lasts for $O(N / t)$ updates, the space for these data structures is $O(N \log N / t)$ bits. 

The third data structure $\mathcal{U}$ is also straightforward to implement. Each time we update $f(x)$ for an element $x$ with $f'(x) = y_i$ for some $i$, we spend $\tilde{O}(t^3)$ time to recover all $x \in f'^{-1}(y_i)$ and then check whether any of them still satisfy $f(x) = y$. If they do not, we remove $i$ from the linked list $\mathcal{U}[y]$ (we can find it in the linked list by keeping a hash table mapping $(y, i)$ to the corresponding list node). This gives a simple implementation of $\mathcal{U}$ that uses $\tilde{O}(N / t)$ bits of space and supports $\tilde{O}(t^3)$-time updates. 

Using these three data structures, we can answer a function-inversion query for an element $y$ as follows. We first check if $y$ has any fresh inverses (using $\mathcal{B}$); if not, we check if $y$ is a heavy hitter for $f_0$, and break into two cases:
\begin{itemize}
\item If it is, we use $\mathcal{U}$ to find some $y_i$ such that not all $f'$-inverses of $y_i$ have been updated (if no such $y_i$ exists, we return $\bot$); we then recover the $f'$-inverses of $y_i$ in $\tilde{O}(t^3)$ time, and check which of them (there must be at least one) still satisfy $f(x) = y$.
\item If it is not, we recover the $f$-inverses of $y$ in $\tilde{O}(t^3)$ time, and check which of them (if any) still satisfy $f(x) = y$. If none of them do, we return $\bot$. 
\end{itemize}
Likewise, using the same approach, we can support an iterator for the inverses of an element $y$ in $\tilde{O}(t^3)$ time per inverse. Finally, rebuilds take $\tilde{O}(N)$ time and $\tilde{O}(N / t)$ bits of space, and can be deamortized exactly as in Section \ref{sec:dynamizing}. This gives Theorem \ref{thm:dynamic-fiat-naor}.

\section{AI Acknowledgement}

In addition to formatting and grammar editing, AI (GPT 5.5 Pro) was used to produce an initial draft of Section \ref{sec:dynamic-unordered-graphs} based on detailed notes provided by the authors. The section was then heavily edited by the authors.

\bibliographystyle{plain}
\bibliography{main}

\end{document}